\documentclass[aps,pra, two column,10pt,superscriptaddress,longbibliography,nofootinbib]{revtex4-2}
\usepackage{etex}
\usepackage{physics}
\usepackage{amsmath,amssymb,amsthm,amstext,mathrsfs}
\usepackage[colorlinks=true,citecolor=blue,urlcolor=blue]{hyperref}
\usepackage{comment}
\usepackage[pdftex]{graphicx}
\usepackage{enumerate}
\usepackage{times,txfonts}
\usepackage{braket}
\usepackage{svg}
\usepackage{color}
\usepackage{natbib}
\usepackage{amsmath,blkarray}
\usepackage{mathtools}
\usepackage{latexsym}
\usepackage{tabularx, booktabs}
\usepackage{graphics,epstopdf}
\usepackage{graphicx, lipsum}
\usepackage{float}
\usepackage{amsfonts,dsfont}
\usepackage{subcaption}
\usepackage{enumerate}
\usepackage{color,soul}

\newcommand{\be}{\begin{equation}}
\newcommand{\ee}{\end{equation}}
\newcommand{\ba}{\begin{eqnarray}}
\newcommand{\ea}{\end{eqnarray}}

\newcommand{\half}{\frac{1}{2}}

\newtheorem{proposition}{Proposition}

\newtheorem{thm}{Theorem}
\newtheorem{Lemma}{Lemma}

\definecolor{ss}{RGB}{250,80,220}
\begin{document}

\title{Contextuality sans  incompatibility in the simplest scenario:  Communication supremacy of a qubit}

\author{Partha Patra}
\email{parthapatra144@gmail.com}
\affiliation{Department of Physics, Indian Institute of Technology Hyderabad, Telengana-502284, India}

\author{Sumit Mukherjee}
\email{mukherjeesumit93@gmail.com}
\affiliation{Korea Research Institute of Standards and Science, Daejeon 34113, South Korea}

\author{A. K. Pan}
\email{akp@phy.iith.ac.in}
\affiliation{Department of Physics, Indian Institute of Technology Hyderabad, Telengana-502284, India}

\begin{abstract}
Conventional wisdom asserts that measurement incompatibility is necessary for revealing the non-locality and contextuality. In contrast, a recent work [\href{https://link.aps.org/doi/10.1103/PhysRevLett.130.230201}{Phys. Rev. Lett. \textbf{130}, 230201 (2023)}] demonstrates the generalized contextuality without measurement incompatibility by using a five-outcome qubit measurement. In this paper, we introduce a two-party prepare-measure communication game involving specific constraints on preparations, and we demonstrate contextuality sans incompatibility in the simplest measurement scenario, requiring only a three-outcome extremal qubit measurement. This contrasts with the aforementioned five-outcome qubit measurement, which can be simulated by an appropriate convex mixture of five three-outcome incompatible qubit measurements.   Furthermore, we illustrate that our result has a prominent implication in information theory. Our communication game can be perceived as a constrained Holevo-Frankle-Weiner (HFW) scenario, as operational restrictions are imposed on preparations. We show that the maximum success probability of the game by using a qubit surpasses that attainable by a $c$-bit, even when shared randomness is a free resource. Consequently, this finding exemplifies the supremacy of a qubit over a $c$-bit within a constrained HFW framework. Thus, alongside offering fresh insights into quantum foundations, our results pave a novel pathway for exploring the efficacy of a qubit in information processing tasks.
\end{abstract}
 
\maketitle

\section{Introduction} No-go theorems play a critical role in elucidating the departure of quantum theory from the classical worldview. Subsequent to Bell's theorem \cite{Bell1964},  the Kochen-Specker (KS) no-go theorem \cite{KS1967} is arguably the most celebrated one. It asserts that certain quantum statistics cannot be reproduced by a noncontextual ontological model that deterministically assigns values to the projectors in co-measurable context - a feature is widely referred to as KS contextuality. The traditional KS contextuality has been subsequently generalized \cite{Spekkens_Contextuality2005} to accommodate unsharp measurements and further extended the notion of contextuality to preparations and transformations.  This form of generalized contextuality has recently been extensively studied \cite{Flatt2022,Mukherjee2022,Spekkens2009,Schmid2022,Saha2019,Pan2019,schmid18dis,schmidgptconPRX,lostaglio2020,dephasingContx23,abhyoudai23,lostaglio20}.

On the other hand, a profound non-classical feature of quantum theory is the existence of measurements that cannot be performed jointly, a feature widely referred to as measurement incompatibility.  It is commonly believed that demonstrating any form of non-classicality, such as nonlocality, steering, or contextuality, necessitates the presence of incompatible measurements. In fact, it has already been established that measurement incompatibility is necessary to reveal nonlocality \cite{wolf09,Bene18,hirsch18} and steering \cite{quintino16,uola16,mukherjee24}.

In contrast to the prevailing consensus, in a recent work, Selby \emph {et al.} \cite{Selby2023} (an elegant refinement of \cite{j.Singh2023}) demonstrated that measurement incompatibility is \emph{not} necessary to reveal generalized contextuality. This is shown in a prepare-measure scenario involving five equally spaced states on an equator of the Bloch sphere and a fixed-input but five-outcome unsharp qubit measurement. They also argued that such a measurement cannot be simulated by post-processing of incompatible sharp measurements. 

In this paper, we address an immediate question. Is it possible to reveal contextuality without measurement incompatibility in a simpler scenario? We answer this question affirmatively for the simplest possible scenario involving only a three-outcome qubit measurement in a prepare-measure communication game. An appropriately defined winning condition coupled with operational constraints on preparations leads to a quantum success probability that surpasses the generalized noncontextual bound. Crucially, our three-outcome qubit measurement cannot be simulated by classical post-processing of other incompatible measurements. However, we observe an intriguing fact that the five-outcome qubit measurement used in \cite{Selby2023} can be simulated through the convex mixture of an appropriate set of five three-outcome measurements. We explicitly show that any pair from the set of five measurements is incompatible.  

Our result has a profound implication in quantum information processing. It necessitates a re-assessment of the asserted equivalence between a $c$-bit and a qubit, as delineated by the celebrated Holevo theorem \cite{Holevo1973} and later generalized by Frankle-Weiner \cite{FW_2015}. The contextual advantage in our communication game, in turn, demonstrates the supremacy of a single qubit over a $c$-bit. At first glance, this appears to be in contradiction with the HFW results \cite{Holevo1973,FW_2015}. However, our communication game incorporates operational restrictions on preparations and, therefore, can be interpreted as a constrained HFW framework. We explicitly prove that unbounded shared randomness cannot diminish qubit communication supremacy. Therefore, beyond its foundational significance, our result paves the way for new research directions to explore information processing tasks in the HFW scenario.

\emph{Generalized contextuality}:- Inspired by the Leibniz principle, the notion of generalized noncontextuality \cite{Spekkens_Contextuality2005} emerges from the fundamental definition of operationally equivalent experimental procedures. The presumed equivalence within an ontological model of the operational theory underpins the formalization of generalized noncontextuality. To describe it more precisely, we take the example of operational quantum theory where the state associated with the preparation procedure $P_{x}$ is represented by the density matrix $\rho_{x} $, and any measurement $M_{y}$ is represented by a set of positive operator-valued measures (POVM) $\{E_{m|y}\}$ corresponding to the outcomes $m$. The operational equivalence between two preparations and measurements in quantum theory are respectively represented as

\begin{equation}\label{equiv1}
    \sum_{x}q_{x} \ \rho_{x}=\sum_{\bar{x}}q_{\bar{x}} \ \rho_{\bar{x}} \ ,  \ \ \sum_{m,y}r_{m,y} E_{m|y}=\sum_{\bar{m},\bar{y}}r_{\bar{m},\bar{y}}E_{\bar{m}|\bar{y}},
\end{equation}

where $q_{x}, q_{\bar{x}},r_{m,y},r_{\bar{y}}$ are appropriate positive real numbers.

An ontological model \cite{Harrigan2010,Spekkens_Contextuality2005,Leifer_2014} of quantum theory posits a set of ontic states $\lambda \in \Lambda$ aiming to provide a realist description of the outcome of any measurement. Quantum preparation $\rho_{x}$ is represented by a probability distribution $\mu(\lambda|\rho_{x})$ on $\lambda$ with $\int_{\Lambda} \mu(\lambda|\rho_{x}) d\lambda =1 $. Given a measurement $M_y$, the ontic state assigns a response function $\xi({m|\lambda,M_{y}})$ corresponding to the outcome $m$, satisfying $\sum_m\xi({m|\lambda,M_{y}})=1, \forall\lambda$. A valid ontological model must reproduce the Born rule $i.e.$, $\sum_{\lambda\in\Lambda}\mu(\lambda|\rho)\xi(m|\lambda,M)=\Tr(\rho E_m)\label{reproduce_born_rule}$.

The assumption of preparation (measurement) noncontextuality in an ontological model demands the following. If two preparations (measurements) procedures are operationally equivalent as in Eq. \eqref{equiv1} then their corresponding ontological description also preserves such equivalences, namely,
\begin{eqnarray}\label{gnc}
&&\sum_{x}q_{x} \ \mu(\lambda|\rho_{x})=\sum_{x}q_{\bar{x}} \ \mu(\lambda|\rho_{\bar{x}}),  \ \forall \lambda \\ \nonumber
    && \sum_{m,y}r_{m,y} \ \xi(m|\lambda,M_{y})=\sum_{\bar{m},\bar{y}}r_{\bar{m},\bar{y}}\xi(\bar{m}|\lambda,M_{\bar{y}}), \ \forall \lambda 
\end{eqnarray}

An ontological model that is preparation and measurement noncontextual is termed as generalized noncontextual. 

\emph{Measurement incompatibility}:-  All measurements in quantum theory cannot be performed jointly. For the case of projective measurement, commutativity of the operators corresponding to the observables ensures the measurement compatibility. However, commutativity is only a sufficient but not a necessary condition for the existence of a joint probability of measurement outcomes. In quantum theory, two measurements $M_i$ and $M_j$ represented by POVMs $\{E_i\}$ and $\{E_j\}$ respectively, are said to be jointly measurable or compatible if there exists a global POVM $G_{i,j}$ that satisfies the properties, i) Positivity: $G_{i,j}\geq 0 \hspace*{0.6cm}\forall i,j$
    ii) Completeness :$\sum_{i,j} G_{i,j}=1$ and
    iii) Reproducibility of marginals: $E_i=\sum_{j}G_{i,j}; \ E_j=\sum_{i}G_{i,j}$. Otherwise, these measurements are called incompatible. Incompatible quantum measurements showcase the non-classicality, and are proven to be a resource in several quantum information processing tasks \cite{Guhne2023}.
    
Note that, measurement incompatibility is typically discussed when at least two POVMs are involved. A POVM representing a single measurement may then be considered as trivially compatible. However, such a POVM may or may  not be simulable \cite{ozmaniac17} by other incompatible POVMs through arbitrary post-processing arrangements. We refer to the former case as pseudo compatibility while the latter case as true compatibility. Distinguishing truly compatible POVMs from those with pseudo compatibility is crucial for our work. We demonstrate the proof of contextuality by using a three-outcome POVM that is truly compatible in the above sense.


\section{Contextuality without incompatibility in simplest scenario}
   We demonstrate the contextuality  without incompatible measurements in the simplest possible scenario by using a single three-outcome qubit measurement. For this, we introduce a prepare-measure communication game played by two parties, Alice and Bob, with constraints on Alice's preparations.
   
 Alice receives six inputs $\{x a\}\in \{00,01,10,11,20,21\}$ from a uniformly random distribution and, accordingly, selects the preparation procedure $\mathbb{P}_{xa}$ and sends the prepared states to Bob. Here $x\in\{0,1,2\}$ and $a\in \{0,1\}$. Bob chooses an appropriate measurement procedure $\mathcal{M}$ that produces three outputs $b\in \{0,1,2\}$. The winning rule of the game is that Bob has to output $b=x\oplus_3a$.
Therefore, the average success probability of winning the game is,
\begin{eqnarray}\label{Avg_Success}
    \mathcal{P}= \frac{1}{6} \sum_{{x\in \{0,1,2\}}\atop a\in\{0,1\}}p(b=x\oplus_3 a|\mathbb{P}_{xa},\mathcal{M})
\end{eqnarray}
Alice can communicate quantum or classical message of the same dimension to Bob. However, Alice's communication abides by specific cryptographic restrictions such that the message should not contain information of $a$ and $x\oplus_3 2a$. We refer to these restrictions as \emph{parity concealment constraints}, formally defined in the following.
\begin{subequations}
  \begin{eqnarray} \label{Parity_Constraint}
    \forall b\& \mathcal{M}\ \ \sum_{x\in \{0,1,2\}} p\left(b|\mathcal{M},\mathbb{P}_{x0}\right)&&= \sum_{x\in \{0,1,2\}} p\left(b|\mathcal{M},\mathbb{P}_{x1}\right) \\ 
    \label{Parity_Constraint2}
    \forall b\& \mathcal{M}\ \ \sum_{x,a:x\oplus_32a=0} p\left(b|\mathcal{M},\mathbb{P}_{xa}\right)&&= \sum_{x,a:x\oplus_32a=1} p(b|\mathcal{M},\mathbb{P}_{xa}) \\ 
    &&=\sum_{x,a:x\oplus_32a=2} p(b|\mathcal{M},\mathbb{P}_{xa})  \nonumber
\end{eqnarray}  
\end{subequations}

In quantum theory, Alice prepares six qubit states $\{\rho_{xa}\}$ corresponding to the preparation procedures $\{\mathbb{P}_{xa}\}$. Upon receiving the qubit from Alice, Bob implements the three-outcome qubit measurement $\mathbb{M}\equiv\{E_b | \ E_b=\alpha_{b} \pi_{b}, \ \sum_{b=0}^{2}E_{b}= \mathds{I} \}$ where $\alpha_b\in [0,1]$. Here, $\pi_{b}$ is the projector of some auxiliary two-outcome measurements. The quantum success probability can then be written as 
\begin{eqnarray}\label{Avg_Succ}
    \mathcal{P}_Q\equiv\frac{1}{6}\sum_{x\in\{0,1,2\}\atop a\in\{0,1\}}\Tr[\rho_{xa} E_{x\oplus_3 a}]
\end{eqnarray}
Note that, to satisfy Eqs. (\ref{Parity_Constraint}) and (\ref{Parity_Constraint2}) the preparations in quantum theory must satisfy $\sum\limits_{x\in\{0,1,2\}}\rho_{x0}$ $=$$\sum\limits_{x\in\{0,1,2\}}\rho_{x1}$ and $\sum\limits_{xa:x\oplus_32a=0}\rho_{xa}$ $=$$\sum\limits_{xa:x\oplus_3 2a=1}\rho_{xa}$$=$$\sum\limits_{xa:x\oplus_3 2a=2}\rho_{xa}$. Following these restrictions, we derive the maximum quantum success probability, as encapsulated in the following theorem.
\begin{Lemma}\label{maxQ}
If Alice communicates a qubit $\rho_{xa}$ and Bob performs a three-outcome measurement $\mathbb{M}$, the maximum success probability $\mathcal{P}_Q^{max}=\frac{1}{3}(1+\frac{\sqrt{3}}{2})$, is attained for $\alpha_b=\frac{2}{3}$, $\forall b$.
\end{Lemma}
The detailed proof is a bit lengthy and is therefore deferred to the Supplemental Material \cite{akp2025}. Here, we provide the states and measurements for the optimal quantum strategy. For $a=0$, Alice's preparation has to be trine-spin qubit states. For example, $\rho_{x0}=\ketbra{\phi_x}{\phi_x}=\frac{1}{2}\left(\mathds{I}+\cos\theta_x \ \sigma_Z+ \sin{\theta_x}\ \sigma_X\right)$, where $\theta_x=\frac{2\pi x}{3}$ with $\ x \in \{0,1,2\}$. For $a=1$, the preparations $\rho_{x1}$ are orthogonal to $\rho_{x0}$ for all $x$ and are represented as $\rho_{x1}=\ketbra{\bar{\phi}_{x\oplus_32}}{\bar{\phi}_{x\oplus_32}}$.   
Note that the constraints on preparations given by Eq. (\ref{Parity_Constraint2}) are naturally satisfied for this set of preparations as $ \frac{1}{2}\sum \limits _{x,a:x\oplus_32a=0}\rho_{xa}=\frac{1}{2}\sum \limits _{x,a:x\oplus_32a=1}\rho_{xa}=\frac{1}{2}\sum \limits _{x,a:x\oplus_32a=2}\rho_{xa}=\frac{\mathds{I}}{2}$ and also Eq.\eqref{Parity_Constraint} as $\frac{1}{3} \sum \limits _{x}\rho_{x0}=\frac{1}{3}\sum \limits _{x}\rho_{x1}=\frac{\mathds{I}}{2}$.

We also proved that the maximum success probability of the communication game is attained for POVM $\{E_b=\alpha_b\pi_b\}$ when $\alpha_b=\frac{2}{3}$,$i.e.$, the trine-spin
POVM. Specifically, $E_{b}=\frac{2}{3}\ketbra{\psi_b}{\psi_b}$, where, $\ketbra{\psi_b}{\psi_b}=\frac{1}{2}\left(\mathds{I}+\cos\theta_b \ \sigma_Z- \sin{\theta_b}\ \sigma_X\right), \theta_b=\frac{2\pi}{3}(\frac{1}{4}+b)$, $\  b \in \{0,1,2\}$. It is crucial to highlight that this POVM is extremal, and cannot be simulated by any classical post-processing of other measurements \cite{D_Ariano_2005}. Therefore, our three-outcome POVM is a truly compatible POVM. 

Here, we note that the operational parity concealment constraints (\ref{Parity_Constraint}) and (\ref{Parity_Constraint2}) are directly reflected in the corresponding ontological model, provided that the model is preparation noncontextual \cite{Spekkens2009}. Consequently, deriving the noncontextual bound for the success probability fundamentally relies on the ontological version of these constraints. To derive this bound, we first prove the following.

\begin{proposition}\label{NO_common_support}
For a generalized noncontextual model that aims to reproduce statistics of a three-outcome qubit POVM $\{E_b=\alpha_b\pi_b\}$, there exists no common ontic state $\lambda_c \in \Lambda$ in support of all three response functions in the set $\{ \xi(b|E_b,\lambda_c)\}$.
\end{proposition}
\begin{proof}
The proof is by contradiction. Consider an elimination game \cite{Perry2015} between Charlie and Debbie where Charlie prepares a state from the set $\mathbb{S}\equiv \{\pi_b=\ketbra{\chi_b}{\chi_b}\}$, here, $b \in \{0,1,2\}\}$, chosen uniformly at random and sends them to Debbie. In every run, Debbie's task is to eliminate one of the states $\ket{\chi_{b'}} : b'\neq b$, that is not sent. It is straightforward to see that the states $\{\ket{\chi_b}\}$ are perfectly anti-distinguishable by the measurement, $\bar{\mathbb{M}}\equiv \{\bar{E_b}\  | \  E_b\bar{E_b}=0 \}$, $i.e.$, 
     \begin{eqnarray}
         \Tr[\ketbra{\chi_b}{\chi_b}\bar{E}_{b'}]=0\ \ \forall b=b'
     \end{eqnarray}
     Therefore, for the outcome $\bar{E}_{b'}$, Debbie perfectly eliminates $\ket{\chi_{b'}}$ and always win the game.
     
Let us analyze the above scenario in an ontological model. Consider that there exists a common $\lambda_c\in\Lambda$ such that $\mu_{\ket{\chi_b}}\big(\lambda_c\big)>0 , \ \forall b$. In the ontological description of the game, Charlie prepares $\mu_{\ket{\chi_b}}\big(\lambda\big)$ for a given quantum state $\ket{\chi_b}$ and sends it to Debbie. In some runs, there is a possibility that the ontic state $\lambda_{c}$ is being sent. For those runs, even after knowing the ontic state $\lambda_{c}$, sometimes Debbie fails to eliminate any of the $\ket{\chi_{b'}}: b'\neq b$ because $\lambda_c$ is common to all $\mu_{\ket{\chi_b}}\big(\lambda\big)$. This leads to a contradiction with quantum case that perfectly eliminates one of the states $\ket{\chi_{b'}}$ in each run of the experiment. Hence, there should not be any common ontic state $\lambda_c$ of the set $\{\mu_{\ket{\chi_b}}\big(\lambda\big)\}_b$.

Here, it is crucial to note that, in an ontological model, 
\begin{eqnarray}
    Supp[\mu_{\ketbra{\chi_b}{\chi_b}}\big(\lambda \big)]\subseteq Supp[\xi(b|\pi_b,\lambda)] \nonumber
\end{eqnarray}
\begin{figure}
    \centering
    \includegraphics[height=40mm, width=87mm,scale=1.5]{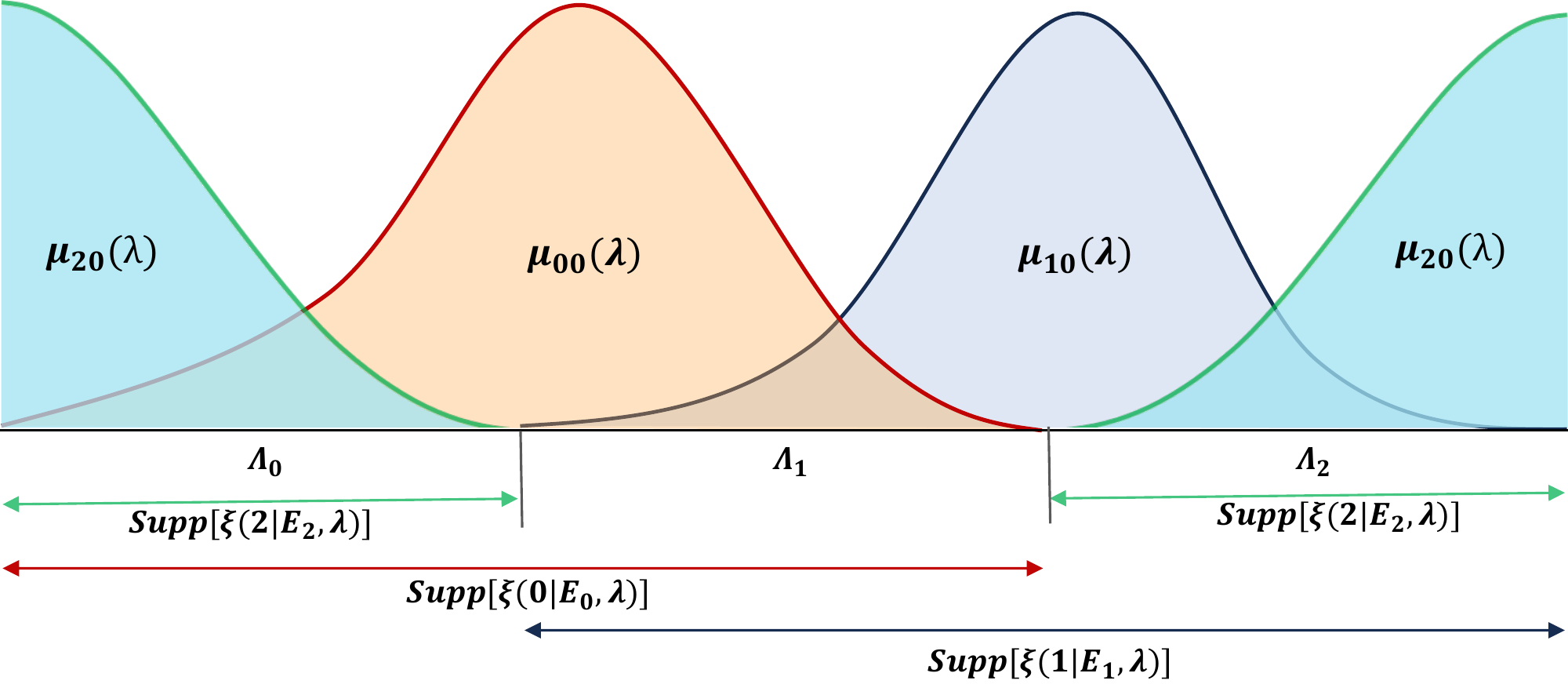}
    \caption{A pictorial illustration of different regions of ontic state space representing preparations $\mu_{x0}(\lambda)$. The whole region is divided into three parts such that each $\lambda$ assigns nonzero values to exactly two  of the response functions $\{\xi(b|E_b,\lambda)\}_{b}$ as proved in Lemma \ref{Each lambda two support}.}
    \label{POVM_Support}
\end{figure}

where, $Supp[p(\lambda)]=\{\lambda \ | \ p(\lambda)>0\}$, $i.e.$, the support of the response functions is bigger than the corresponding support of the state space of preparations \cite{Harrigan2007ontologicalmodelsinterpretationcontextuality}. In this instance, we recall a well-established result \cite{Leifer_Marony(2013), Pan_2021} that a preparation noncontextual model is  maximally $\psi$-epistemic one. In such a model, the ontic state space overlap of two distinct preparations must be the same as the Hilbert space overlap between the quantum states to which these preparations correspond \cite{Maroney2012}. In our case, as $E_b=\alpha_b\pi_b$, we can then write $Supp[\xi(b|\pi_b,\lambda)] =Supp[\xi(b|E_b,\lambda)]$. Therefore, in a preparation noncontextual model corresponding to  measurement $\mathbb{M}$,
\begin{eqnarray}
    Supp[\xi(b|E_b,\lambda)]=Supp[\mu_{\ket{\chi_b}}\big(\lambda \big)]
\end{eqnarray}

Since $\{\mu_{\ket{\chi_b}}\big(\lambda\big)\}_b$ has no common support, therefore, there should be no common support of all response functions in the set $\{\xi(b|E_b,\lambda)\}$ in a preparation noncontextual model.
\end{proof}

\begin{proposition}\label{No Deterministic model}
A generalized noncontextual model that aims to reproduce the statistics of a three-outcome qubit POVM $\{E_b=\alpha_b\pi_b\},\alpha_b>0$ cannot be deterministic, and the corresponding response functions $\{\xi(b|E_b,\lambda)\}$ must be bounded as $0\leq\xi(b|E_b,\lambda)\leq \alpha_b\ , \forall b$. 
\end{proposition}
\begin{proof}
Consider again the measurement $\mathbb{M}\equiv\big\{E_{b}|E_{b}=\alpha_b\pi_b \big\}$. In an ontological model, the completeness relation $\sum \limits_{b\in\{0,1,2\}}\alpha_b\pi_b=\mathds{I}$ implies $\sum \limits_{b\in\{0,1,2\}}\alpha_b \xi(b|\pi_b, \lambda)=1$. Moreover, since $\frac{1}{2}(\pi_{b}+\pi_b^{\perp})=\frac{1}{2}\sum \limits_{b\in\{0,1,2\}}\alpha_b\pi_b=\frac{\mathds{I}}{2}$ constitutes a valid operational equivalence 
 defined in Eq.\eqref{equiv1} for the POVM element $\frac{\mathds{I}}{2}$. Therefore, the measurement noncontextuality assumption of Eq. \eqref{gnc} trivially dictates $\xi(b|\pi_{b},\lambda)+\xi(b|\pi_{b}^{\perp},\lambda)=1$. Now, it is straightforward to check that $\xi(b|\lambda,\pi_b)\in\{0,1\}\ \forall b$ does not satisfy $\sum \limits_{b}\alpha_b=2$ for $\alpha_b>0, \forall b$. Therefore, for any $\lambda$, an ontological model cannot assign a deterministic value to all the response functions of the set $\{\xi(b|\pi_b,\lambda)\}$. 
    
     Although for a given $\lambda$ the three response functions associated with the projectors cannot take deterministic values, they are still bounded by $0\leq\xi(b|\pi_b,\lambda)\leq 1$ $\forall b \in \{0,1,2\}$. Since $E_{b}=\alpha_{b}\pi_{b}$, the corresponding response functions must be bounded as $0\leq\xi(b|E_b,\lambda)\leq \alpha_b\ \forall b \in \{0,1,2\}$.
\end{proof}

\begin{Lemma}\label{Each lambda two support}
    A generalized noncontextual model reproducing the statistics of a three-outcome qubit POVM $\{E_b=\alpha_b\pi_b\}$, for every ontic state $\lambda \in \Lambda$, exactly two of the response functions from the set $\{\xi(b|E_b,\lambda)\}$ assume non-zero values.
\end{Lemma}

\textit{Proof:} From Proposition \ref{No Deterministic model} it is clear that as the set of response functions has to satisfy the ontological version of the completeness relation, each $\lambda\in \Lambda$ can be in the support of either two or all three response functions in the set $\{\xi(b|E_b,\lambda)\}$ corresponding $\{E_b=\alpha_b\pi_b\}$, but \emph{cannot} be in the support of only one response function. Proposition \ref{NO_common_support} implies that any $\lambda\in \Lambda$ \emph{cannot} be in the supports of all the three response functions in the set $\{\xi(b|E_b,\lambda)\}$. Therefore, the conjunction of the aforementioned propositions directly leads to Lemma \ref{Each lambda two support}.


\begin{thm}\label{general_th}
The quantum advantage over the noncontextual model can be achieved for any arbitrary value of $\{\alpha_b\}$, except for two extremes when one of the $\alpha_b$  is $0$ or $1$. The maximum quantum advantage is achieved for the trine-spin measurement ($\alpha_{b}=\frac{2}{3}$, $\forall b$).
\end{thm}

\textit{Sketch of the proof:-} An explicit proof is deferred to the Supplemental Material \cite{akp2025}. Here we provide a sketch of it. 

In an ontological model, the average success probability defined in Eq.\eqref{Avg_Success} can be written as
\begin{eqnarray}\label{NC_Avg_Succ}
     \mathcal{P}_{ont}\equiv \frac{1}{6}\sum_{\lambda\in\Lambda}\sum_{x=0}^{2}[\mu_{x0}(\lambda)\xi\left(x|\lambda\right)+\mu_{x1}(\lambda)\xi\left(x \oplus_{3} 1|\lambda\right)]
\end{eqnarray}
Following Lemma \ref{Each lambda two support}, we divide the entire ontic state space $\Lambda$ into three disjoint parts $\Lambda_0,\Lambda_1$ and $\Lambda_2$, so that $\Lambda_0\cup\Lambda_1\cup\Lambda_2=\Lambda$ as illustrated in Fig.\ref{POVM_Support}. The proposition \ref{No Deterministic model} bounds the response functions by $0\leq\xi(b|E_b,\lambda)\leq \alpha_b\ \forall b $. To derive the maximum noncontextual value $\mathcal{P}_{NC}$, we take the maximum value corresponding to the response functions in the individual regions, and finally impose the parity concealment constraints. A straightforward optimization yields,
\begin{equation}
    \mathcal{P}_{NC}=\max_{\mu_{xa}(\lambda),\xi(b|E_b,\lambda)} \mathcal{P}_{ont} \leq \frac{7}{12}
\end{equation}
Such a value is obtained in two extreme cases (Fig. \ref{Q-NC}), $\alpha_{0}=1$ and $\alpha_{1}=\alpha_{2}=\frac{1}{2}$, and $\alpha_{0}= 0$ and $\alpha_{1}=\alpha_{2}=1$. However, in those two cases, no quantum advantage is obtained, as both quantum and noncontextual values are the same, i.e., $\mathcal{P}_{Q}=\mathcal{P}_{NC}=7/12\approx0.58$. We note that those two cases effectively correspond to a two-outcome measurement.

We analytically prove \cite{akp2025} that, corresponding to the optimal quantum strategy for $\alpha_{b}=\frac{2}{3}$, $\forall b$, the maximum success probability of the communication game in a noncontextual model is $\frac{1}{2}$, thus providing the maximum quantum advantage.  In Fig. \ref{Q-NC}, we  compare the maximum quantum and noncontextual values of the success probability for a given value of $\alpha_{0}$ with $\alpha_1 =\alpha_2$.

\begin{figure}
     \includegraphics[height=50mm, width=85mm,scale=1.5]{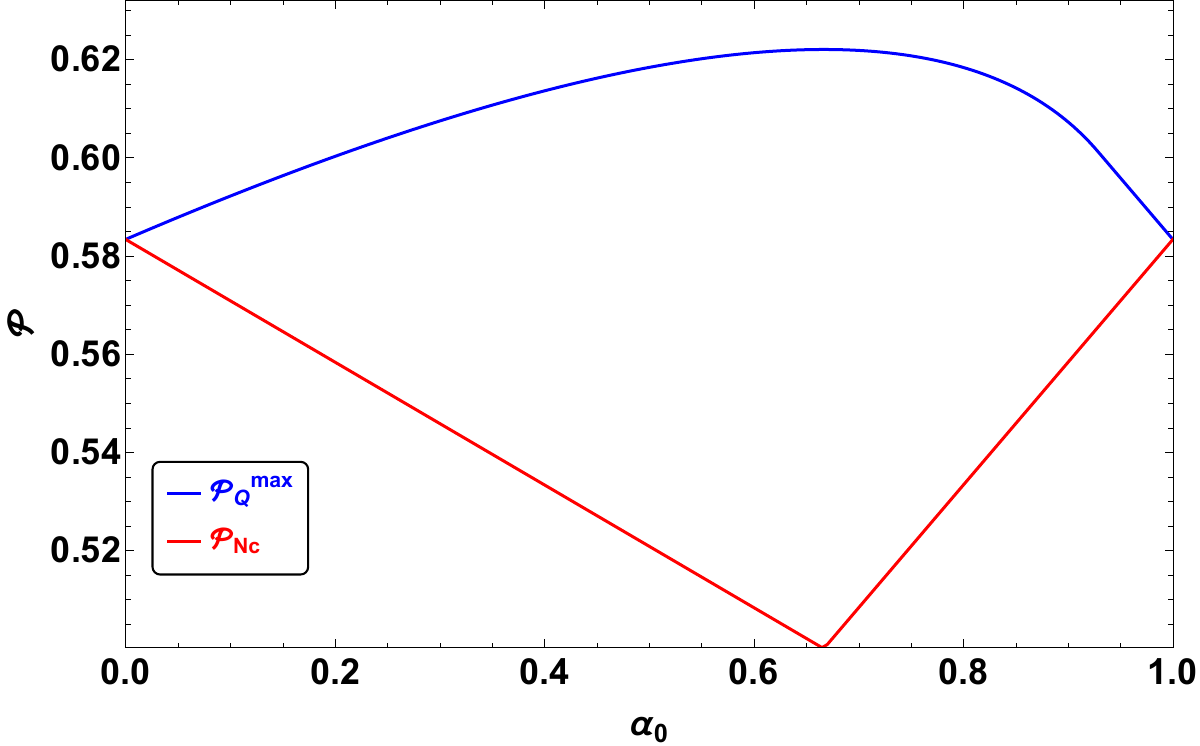}
    \caption{The maximum success probabilities $\mathcal{P}_Q$ and  $\mathcal{P}_{NC}$ are plotted against $\alpha_0$, with $\alpha_1=\alpha_2$.}
    \label{Q-NC}
\end{figure}

\emph{Remarks on the relevant work}:- In their proof of contextuality without incompatibility, Selby \emph{et al.}\cite{Selby2023} used a single five-outcome qubit measurement. We show that such POVMs can be successfully simulated by a suitably chosen set $\mathcal{S}_M\equiv\{M^0,M^1,M^2,M^3,M^4\}$ of three-outcome extremal qubit measurements. We also prove that any of the two measurements in the set $\mathcal{S}_M$ are incompatible. For this, we first prove the following.

\begin{Lemma}\label{Simulating_by_trine}
     There exists a set of three-outcome qubit measurements whose suitable convexification can simulate any arbitrary odd $n$-outcome  qubit measurement  $\mathbb{M}_{n}\equiv  \{E_{k}=\frac{2}{n}\pi_k, \  k\in \{0,1,...,n-1\}\}$  where the projector $\pi_k=\frac{1}{2}\left( \mathds{I}+\cos{\frac{2\pi k}{n}}\sigma_z + \sin{\frac{2\pi k}{n}} \sigma_x\right)$ lies in the great circle of the Bloch sphere with $n\geq 3 $.
     \end{Lemma} 
 The proof of Lemma \ref{Simulating_by_trine} is placed in the Supplemental Material \cite{akp2025}. Clearly, for $n=5$, the five-outcome qubit POVM used in \cite{Selby2023} can be reproduced by a set $\mathcal{S}_5\equiv\{M^0,M^1,M^2,M^3,M^4\}$ where $M^o$s are three-outcome extremal qubit POVMs. The explicit form of $\mathcal{S}_5$ is placed in the Supplemental Material \cite{akp2025}. Using the results of ref.\cite{Heinosaari2018,CarmeliIW2019} we also show that any pair of $M^o$s  are incompatible. We can then say that the POVM used in \cite{Selby2023} is pseudo-compatible. In contrast, the three-outcome extremal POVM used in this paper to achieve the maximum quantum advantage cannot be simulated in this way and is truly compatible.

\section{Implication to the HFW Scenario} The contextual quantum advantage of the communication game demonstrated here falls under the premise of the HFW scenario. The Holevo theorem \cite{Holevo1973} posits that a $q$-dit communication is not advantageous than a classical $d$-level system, as the mutual information 
is equivalent in both cases. Later, Frenkel and Weiner\cite{FW_2015} improved this result by showing that any $q$-dit statistics can be simulable by using a classical $d$-bit in the Holevo scenario, provided that the classical shared randomness is a free resource. 
In contrast, in Theorem \ref{thm2}, we demonstrate the supremacy of a qubit in a constrained HFW framework.  
\begin{thm}
\label{thm2}
    In a HFW scenario supplemented with parity concealment constraints, qubit communication overpowers $c$-bit, even in the presence of unbounded shared randomness.
\end{thm}
The proof of the theorem is somewhat mathematical and is deferred to the Supplemental Material \cite{akp2025}. We rigorously prove that any $c$-bit communication strategy respecting the parity concealment constraints yields a maximum success probability $\left(\mathcal{P}_c\right)^{max}=7/12\approx 0.58$ of our communication game.  Now, if the success probability is considered as a figure of merit, at first glance, our results might seem incongruent with the HFW theorem. However, the game is played in a constrained scenario, as the encoding strategy adheres to the constraints specified in Eqs. (\ref{Parity_Constraint}) and (\ref{Parity_Constraint2}). 
The imposition of the parity concealment constraints impacts the classical and quantum strategies quite differently, eventually leading to qubit supremacy. 


\emph{Remark-1:-} An immediate question could be the following. Can the quantum advantage be demonstrated in the entanglement-assisted $c$-bit communication scenario, instead of qubit communication? The answer to this question is affirmative and is quite straightforward to prove.  Alice performs three projective measurements on her shared qubit and, depending on the outcome, sends one $c$-bit so that the steered states to Bob follow the parity concealment constraints. Using this $c$-bit, Bob chooses to flip or not to flip his shared qubit according to their predetermined strategy and then performs the three-outcome qubit measurement. In this way, they can obtain the same success probability as $\mathcal{P}_Q^{max}$.

\emph{Remark-2:-} In the absence of classical shared randomness, a form of qubit supremacy over c-bit is demonstrated \cite{Patra2024,Zhonghua2023,chen_banik2024}. Note that this constraint is quite stringent. Even allowing non-locality in an ontological model is not enough to simulate the quantum correlation in the absence of classical shared randomness. 
In this paper, we impose a less stringent constraint on Alice's encoding strategy to reveal the supremacy of qubit communication.

\emph{Remark-3:-} Communication advantage of qubit has already been explored in random-access-code\cite{Wiesner1983RAC,Ambainis1999RAC,Ambainis2002RAC} where Bob possesses seed randomness to choose his measurements. A constrained random-access-code scenario was proposed \cite{Spekkens2009} by imposing a parity-oblivious restriction on Alice's preparation. 
This constraint is sufficiently potent to decrease the classical success probability and facilitates a quantum advantage \cite{ghorai2018}.

\emph{Remark-4:-} A recent result \cite{saha2023pra} demonstrates that compatible quantum measurements cannot provide an advantage over classical theory in any communication scenario. This seems to contradict the findings of this paper and of \cite{Selby2023}. However, the key difference is that the communication game in \cite{saha2023pra} does not involve operational restriction in preparation.

\section{Summary and Discussion} In sum, we demonstrated a novel proof of generalized contextuality devoid of measurement incompatibility in the simplest possible scenario featuring merely a three-outcome qubit measurement. The first proof of a similar kind was demonstrated in \cite{Selby2023}, but uses a five-outcome qubit measurement. We showed that the measurement used in \cite{Selby2023} can be simulated by implementing a set of five three-outcome incompatible measurements. However, this does not undermine their proof, as there are many other practical ways to implement a specific measurement corresponding to the POVM concerned. In contrast, the three-outcome extremal qubit POVM used in this paper cannot be simulated.

The communication game introduced in this paper, when considered from the perspective of information processing, aligns with the HFW framework \cite{Holevo1973, FW_2015}. However, we imposed an operational constraint on Alice's preparation. Within this constrained HFW framework, we demonstrated the communication supremacy of a qubit over a $c$-bit, contingent upon both classical and quantum theories complying with the specified operational restrictions.

 The revelation of contextual advantage without incompatible measurements has defenestrated the belief that measurement incompatibility constitutes the most fundamental non-classical aspect of quantum measurements. This invites a deeper inquiry into what non-classicality is essential to manifest such a quantum advantage. We have demonstrated (detailed in Supplemental Material \cite{akp2025}) that POVMs capable of detecting state coherence \cite{plenioprl19} are the crucial resource to demonstrate contextuality without incompatibility in the discussed scenario. An in-depth investigation into the characterization of these POVMs is forthcoming in a future publication. A recent work \cite{jonte25} may also shed light towards this direction.

The quantum advantage in a constrained prepare-measure scenario has already been exemplified through the constraints imposed on preparations, such as parity-obliviousness \cite{Spekkens2009}, energy content \cite{VanHimbeeck2017semidevice}, and information content \cite{Tavakoli2020informationally}. It is important to observe that the quantum advantage in such a scenario emerges because of the differing impacts that specific constraints have on quantum theory compared to classical theory. In light of our results, it would be of interest to introduce various other forms of constraint to demonstrate quantum supremacy in the HFW scenario. We also conjecture that for any fixed measurement empowered with state coherence detection capability, there exist suitable constraints in preparation so that quantum supremacy in the HFW scenario can be demonstrated. Investigations following this path could be an exciting direction for future research.

 \emph{Acknowledgements:-}
P.P. acknowledges funding from the University Grants Commission (NTA Ref. No.-231610049800), Govt. of India. S.M acknowledges the local hospitality from the grant IITH/SG160 of IIT Hyderabad, India. A.K.P. acknowledges the support from the Research Grant
 SERB/CRG/2021/004258, Government of India.

\bibliography{ref}

\begin{thebibliography}{50}%
\makeatletter
\providecommand \@ifxundefined [1]{%
 \@ifx{#1\undefined}
}%
\providecommand \@ifnum [1]{%
 \ifnum #1\expandafter \@firstoftwo
 \else \expandafter \@secondoftwo
 \fi
}%
\providecommand \@ifx [1]{%
 \ifx #1\expandafter \@firstoftwo
 \else \expandafter \@secondoftwo
 \fi
}%
\providecommand \natexlab [1]{#1}%
\providecommand \enquote  [1]{``#1''}%
\providecommand \bibnamefont  [1]{#1}%
\providecommand \bibfnamefont [1]{#1}%
\providecommand \citenamefont [1]{#1}%
\providecommand \href@noop [0]{\@secondoftwo}%
\providecommand \href [0]{\begingroup \@sanitize@url \@href}%
\providecommand \@href[1]{\@@startlink{#1}\@@href}%
\providecommand \@@href[1]{\endgroup#1\@@endlink}%
\providecommand \@sanitize@url [0]{\catcode `\\12\catcode `\$12\catcode `\&12\catcode `\#12\catcode `\^12\catcode `\_12\catcode `\%12\relax}%
\providecommand \@@startlink[1]{}%
\providecommand \@@endlink[0]{}%
\providecommand \url  [0]{\begingroup\@sanitize@url \@url }%
\providecommand \@url [1]{\endgroup\@href {#1}{\urlprefix }}%
\providecommand \urlprefix  [0]{URL }%
\providecommand \Eprint [0]{\href }%
\providecommand \doibase [0]{https://doi.org/}%
\providecommand \selectlanguage [0]{\@gobble}%
\providecommand \bibinfo  [0]{\@secondoftwo}%
\providecommand \bibfield  [0]{\@secondoftwo}%
\providecommand \translation [1]{[#1]}%
\providecommand \BibitemOpen [0]{}%
\providecommand \bibitemStop [0]{}%
\providecommand \bibitemNoStop [0]{.\EOS\space}%
\providecommand \EOS [0]{\spacefactor3000\relax}%
\providecommand \BibitemShut  [1]{\csname bibitem#1\endcsname}%
\let\auto@bib@innerbib\@empty
\bibitem [{\citenamefont {Bell}(1964)}]{Bell1964}%
  \BibitemOpen
  \bibfield  {author} {\bibinfo {author} {\bibfnamefont {J.~S.}\ \bibnamefont {Bell}},\ }\bibfield  {title} {\bibinfo {title} {On the einstein podolsky rosen paradox},\ }\href {https://doi.org/10.1103/PhysicsPhysiqueFizika.1.195} {\bibfield  {journal} {\bibinfo  {journal} {Physics Physique Fizika}\ }\textbf {\bibinfo {volume} {1}},\ \bibinfo {pages} {195} (\bibinfo {year} {1964})}\BibitemShut {NoStop}%
\bibitem [{\citenamefont {Kochen}\ and\ \citenamefont {Specker}(1967)}]{KS1967}%
  \BibitemOpen
  \bibfield  {author} {\bibinfo {author} {\bibfnamefont {S.}~\bibnamefont {Kochen}}\ and\ \bibinfo {author} {\bibfnamefont {E.~P.}\ \bibnamefont {Specker}},\ }\bibfield  {title} {\bibinfo {title} {The problem of hidden variables in quantum mechanics},\ }\href {http://www.jstor.org/stable/24902153} {\bibfield  {journal} {\bibinfo  {journal} {Journal of Mathematics and Mechanics}\ }\textbf {\bibinfo {volume} {17}},\ \bibinfo {pages} {59} (\bibinfo {year} {1967})}\BibitemShut {NoStop}%
\bibitem [{\citenamefont {Spekkens}(2005)}]{Spekkens_Contextuality2005}%
  \BibitemOpen
  \bibfield  {author} {\bibinfo {author} {\bibfnamefont {R.~W.}\ \bibnamefont {Spekkens}},\ }\bibfield  {title} {\bibinfo {title} {Contextuality for preparations, transformations, and unsharp measurements},\ }\href {https://doi.org/10.1103/PhysRevA.71.052108} {\bibfield  {journal} {\bibinfo  {journal} {Phys. Rev. A}\ }\textbf {\bibinfo {volume} {71}},\ \bibinfo {pages} {052108} (\bibinfo {year} {2005})}\BibitemShut {NoStop}%
\bibitem [{\citenamefont {Flatt}\ \emph {et~al.}(2022)\citenamefont {Flatt}, \citenamefont {Lee}, \citenamefont {Carceller}, \citenamefont {Brask},\ and\ \citenamefont {Bae}}]{Flatt2022}%
  \BibitemOpen
  \bibfield  {author} {\bibinfo {author} {\bibfnamefont {K.}~\bibnamefont {Flatt}}, \bibinfo {author} {\bibfnamefont {H.}~\bibnamefont {Lee}}, \bibinfo {author} {\bibfnamefont {C.~R.~I.}\ \bibnamefont {Carceller}}, \bibinfo {author} {\bibfnamefont {J.~B.}\ \bibnamefont {Brask}},\ and\ \bibinfo {author} {\bibfnamefont {J.}~\bibnamefont {Bae}},\ }\bibfield  {title} {\bibinfo {title} {Contextual advantages and certification for maximum-confidence discrimination},\ }\href {https://doi.org/10.1103/PRXQuantum.3.030337} {\bibfield  {journal} {\bibinfo  {journal} {PRX Quantum}\ }\textbf {\bibinfo {volume} {3}},\ \bibinfo {pages} {030337} (\bibinfo {year} {2022})}\BibitemShut {NoStop}%
\bibitem [{\citenamefont {Mukherjee}\ \emph {et~al.}(2022)\citenamefont {Mukherjee}, \citenamefont {Naonit},\ and\ \citenamefont {Pan}}]{Mukherjee2022}%
  \BibitemOpen
  \bibfield  {author} {\bibinfo {author} {\bibfnamefont {S.}~\bibnamefont {Mukherjee}}, \bibinfo {author} {\bibfnamefont {S.}~\bibnamefont {Naonit}},\ and\ \bibinfo {author} {\bibfnamefont {A.~K.}\ \bibnamefont {Pan}},\ }\bibfield  {title} {\bibinfo {title} {Discriminating three mirror-symmetric states with a restricted contextual advantage},\ }\href {https://doi.org/10.1103/PhysRevA.106.012216} {\bibfield  {journal} {\bibinfo  {journal} {Phys. Rev. A}\ }\textbf {\bibinfo {volume} {106}},\ \bibinfo {pages} {012216} (\bibinfo {year} {2022})}\BibitemShut {NoStop}%
\bibitem [{\citenamefont {Spekkens}\ \emph {et~al.}(2009)\citenamefont {Spekkens}, \citenamefont {Buzacott}, \citenamefont {Keehn}, \citenamefont {Toner},\ and\ \citenamefont {Pryde}}]{Spekkens2009}%
  \BibitemOpen
  \bibfield  {author} {\bibinfo {author} {\bibfnamefont {R.~W.}\ \bibnamefont {Spekkens}}, \bibinfo {author} {\bibfnamefont {D.~H.}\ \bibnamefont {Buzacott}}, \bibinfo {author} {\bibfnamefont {A.~J.}\ \bibnamefont {Keehn}}, \bibinfo {author} {\bibfnamefont {B.}~\bibnamefont {Toner}},\ and\ \bibinfo {author} {\bibfnamefont {G.~J.}\ \bibnamefont {Pryde}},\ }\bibfield  {title} {\bibinfo {title} {Preparation contextuality powers parity-oblivious multiplexing},\ }\href {https://doi.org/10.1103/PhysRevLett.102.010401} {\bibfield  {journal} {\bibinfo  {journal} {Phys. Rev. Lett.}\ }\textbf {\bibinfo {volume} {102}},\ \bibinfo {pages} {010401} (\bibinfo {year} {2009})}\BibitemShut {NoStop}%
\bibitem [{\citenamefont {Schmid}\ \emph {et~al.}(2022)\citenamefont {Schmid}, \citenamefont {Du}, \citenamefont {Selby},\ and\ \citenamefont {Pusey}}]{Schmid2022}%
  \BibitemOpen
  \bibfield  {author} {\bibinfo {author} {\bibfnamefont {D.}~\bibnamefont {Schmid}}, \bibinfo {author} {\bibfnamefont {H.}~\bibnamefont {Du}}, \bibinfo {author} {\bibfnamefont {J.~H.}\ \bibnamefont {Selby}},\ and\ \bibinfo {author} {\bibfnamefont {M.~F.}\ \bibnamefont {Pusey}},\ }\bibfield  {title} {\bibinfo {title} {Uniqueness of noncontextual models for stabilizer subtheories},\ }\href {https://doi.org/10.1103/PhysRevLett.129.120403} {\bibfield  {journal} {\bibinfo  {journal} {Phys. Rev. Lett.}\ }\textbf {\bibinfo {volume} {129}},\ \bibinfo {pages} {120403} (\bibinfo {year} {2022})}\BibitemShut {NoStop}%
\bibitem [{\citenamefont {Saha}\ and\ \citenamefont {Chaturvedi}(2019)}]{Saha2019}%
  \BibitemOpen
  \bibfield  {author} {\bibinfo {author} {\bibfnamefont {D.}~\bibnamefont {Saha}}\ and\ \bibinfo {author} {\bibfnamefont {A.}~\bibnamefont {Chaturvedi}},\ }\bibfield  {title} {\bibinfo {title} {Preparation contextuality as an essential feature underlying quantum communication advantage},\ }\href {https://doi.org/10.1103/PhysRevA.100.022108} {\bibfield  {journal} {\bibinfo  {journal} {Phys. Rev. A}\ }\textbf {\bibinfo {volume} {100}},\ \bibinfo {pages} {022108} (\bibinfo {year} {2019})}\BibitemShut {NoStop}%
\bibitem [{\citenamefont {Pan}(2019)}]{Pan2019}%
  \BibitemOpen
  \bibfield  {author} {\bibinfo {author} {\bibfnamefont {A.~K.}\ \bibnamefont {Pan}},\ }\bibfield  {title} {\bibinfo {title} {Revealing universal quantum contextuality through communication games},\ }\href {https://doi.org/10.1038/s41598-019-53701-5} {\bibfield  {journal} {\bibinfo  {journal} {Scientific Reports}\ }\textbf {\bibinfo {volume} {9}},\ \bibinfo {pages} {17631} (\bibinfo {year} {2019})}\BibitemShut {NoStop}%
\bibitem [{\citenamefont {Schmid}\ and\ \citenamefont {Spekkens}(2018)}]{schmid18dis}%
  \BibitemOpen
  \bibfield  {author} {\bibinfo {author} {\bibfnamefont {D.}~\bibnamefont {Schmid}}\ and\ \bibinfo {author} {\bibfnamefont {R.~W.}\ \bibnamefont {Spekkens}},\ }\bibfield  {title} {\bibinfo {title} {Contextual advantage for state discrimination},\ }\href {https://doi.org/10.1103/PhysRevX.8.011015} {\bibfield  {journal} {\bibinfo  {journal} {Phys. Rev. X}\ }\textbf {\bibinfo {volume} {8}},\ \bibinfo {pages} {011015} (\bibinfo {year} {2018})}\BibitemShut {NoStop}%
\bibitem [{\citenamefont {Schmid}\ \emph {et~al.}(2021)\citenamefont {Schmid}, \citenamefont {Selby}, \citenamefont {Wolfe}, \citenamefont {Kunjwal},\ and\ \citenamefont {Spekkens}}]{schmidgptconPRX}%
  \BibitemOpen
  \bibfield  {author} {\bibinfo {author} {\bibfnamefont {D.}~\bibnamefont {Schmid}}, \bibinfo {author} {\bibfnamefont {J.~H.}\ \bibnamefont {Selby}}, \bibinfo {author} {\bibfnamefont {E.}~\bibnamefont {Wolfe}}, \bibinfo {author} {\bibfnamefont {R.}~\bibnamefont {Kunjwal}},\ and\ \bibinfo {author} {\bibfnamefont {R.~W.}\ \bibnamefont {Spekkens}},\ }\bibfield  {title} {\bibinfo {title} {Characterization of noncontextuality in the framework of generalized probabilistic theories},\ }\href {https://doi.org/10.1103/PRXQuantum.2.010331} {\bibfield  {journal} {\bibinfo  {journal} {PRX Quantum}\ }\textbf {\bibinfo {volume} {2}},\ \bibinfo {pages} {010331} (\bibinfo {year} {2021})}\BibitemShut {NoStop}%
\bibitem [{\citenamefont {Lostaglio}\ and\ \citenamefont {Senno}(2020)}]{lostaglio2020}%
  \BibitemOpen
  \bibfield  {author} {\bibinfo {author} {\bibfnamefont {M.}~\bibnamefont {Lostaglio}}\ and\ \bibinfo {author} {\bibfnamefont {G.}~\bibnamefont {Senno}},\ }\bibfield  {title} {\bibinfo {title} {Contextual advantage for state-dependent cloning},\ }\bibfield  {journal} {\bibinfo  {journal} {Quantum}\ }\href {https://doi.org/10.22331/q-2020-04-27-258} {10.22331/q-2020-04-27-258} (\bibinfo {year} {2020})\BibitemShut {NoStop}%
\bibitem [{\citenamefont {Rossi}\ \emph {et~al.}(2023)\citenamefont {Rossi}, \citenamefont {Schmid}, \citenamefont {Selby},\ and\ \citenamefont {Sainz}}]{dephasingContx23}%
  \BibitemOpen
  \bibfield  {author} {\bibinfo {author} {\bibfnamefont {V.~P.}\ \bibnamefont {Rossi}}, \bibinfo {author} {\bibfnamefont {D.}~\bibnamefont {Schmid}}, \bibinfo {author} {\bibfnamefont {J.~H.}\ \bibnamefont {Selby}},\ and\ \bibinfo {author} {\bibfnamefont {A.~B.}\ \bibnamefont {Sainz}},\ }\bibfield  {title} {\bibinfo {title} {Contextuality with vanishing coherence and maximal robustness to dephasing},\ }\href {https://doi.org/10.1103/PhysRevA.108.032213} {\bibfield  {journal} {\bibinfo  {journal} {Phys. Rev. A}\ }\textbf {\bibinfo {volume} {108}},\ \bibinfo {pages} {032213} (\bibinfo {year} {2023})}\BibitemShut {NoStop}%
\bibitem [{\citenamefont {S.~S.}\ \emph {et~al.}(2023)\citenamefont {S.~S.}, \citenamefont {Mukherjee},\ and\ \citenamefont {Pan}}]{abhyoudai23}%
  \BibitemOpen
  \bibfield  {author} {\bibinfo {author} {\bibfnamefont {A.}~\bibnamefont {S.~S.}}, \bibinfo {author} {\bibfnamefont {S.}~\bibnamefont {Mukherjee}},\ and\ \bibinfo {author} {\bibfnamefont {A.~K.}\ \bibnamefont {Pan}},\ }\bibfield  {title} {\bibinfo {title} {Robust certification of unsharp instruments through sequential quantum advantages in a prepare-measure communication game},\ }\href {https://doi.org/10.1103/PhysRevA.107.012411} {\bibfield  {journal} {\bibinfo  {journal} {Phys. Rev. A}\ }\textbf {\bibinfo {volume} {107}},\ \bibinfo {pages} {012411} (\bibinfo {year} {2023})}\BibitemShut {NoStop}%
\bibitem [{\citenamefont {Lostaglio}(2020)}]{lostaglio20}%
  \BibitemOpen
  \bibfield  {author} {\bibinfo {author} {\bibfnamefont {M.}~\bibnamefont {Lostaglio}},\ }\bibfield  {title} {\bibinfo {title} {Certifying quantum signatures in thermodynamics and metrology via contextuality of quantum linear response},\ }\href {https://doi.org/10.1103/PhysRevLett.125.230603} {\bibfield  {journal} {\bibinfo  {journal} {Phys. Rev. Lett.}\ }\textbf {\bibinfo {volume} {125}},\ \bibinfo {pages} {230603} (\bibinfo {year} {2020})}\BibitemShut {NoStop}%
\bibitem [{\citenamefont {Wolf}\ \emph {et~al.}(2009)\citenamefont {Wolf}, \citenamefont {Perez-Garcia},\ and\ \citenamefont {Fernandez}}]{wolf09}%
  \BibitemOpen
  \bibfield  {author} {\bibinfo {author} {\bibfnamefont {M.~M.}\ \bibnamefont {Wolf}}, \bibinfo {author} {\bibfnamefont {D.}~\bibnamefont {Perez-Garcia}},\ and\ \bibinfo {author} {\bibfnamefont {C.}~\bibnamefont {Fernandez}},\ }\bibfield  {title} {\bibinfo {title} {Measurements incompatible in quantum theory cannot be measured jointly in any other no-signaling theory},\ }\href {https://doi.org/10.1103/PhysRevLett.103.230402} {\bibfield  {journal} {\bibinfo  {journal} {Phys. Rev. Lett.}\ }\textbf {\bibinfo {volume} {103}},\ \bibinfo {pages} {230402} (\bibinfo {year} {2009})}\BibitemShut {NoStop}%
\bibitem [{\citenamefont {Bene}\ and\ \citenamefont {Vértesi}(2018)}]{Bene18}%
  \BibitemOpen
  \bibfield  {author} {\bibinfo {author} {\bibfnamefont {E.}~\bibnamefont {Bene}}\ and\ \bibinfo {author} {\bibfnamefont {T.}~\bibnamefont {Vértesi}},\ }\bibfield  {title} {\bibinfo {title} {Measurement incompatibility does not give rise to bell violation in general},\ }\href {https://doi.org/10.1088/1367-2630/aa9ca3} {\bibfield  {journal} {\bibinfo  {journal} {New Journal of Physics}\ }\textbf {\bibinfo {volume} {20}},\ \bibinfo {pages} {013021} (\bibinfo {year} {2018})}\BibitemShut {NoStop}%
\bibitem [{\citenamefont {Hirsch}\ \emph {et~al.}(2018)\citenamefont {Hirsch}, \citenamefont {Quintino},\ and\ \citenamefont {Brunner}}]{hirsch18}%
  \BibitemOpen
  \bibfield  {author} {\bibinfo {author} {\bibfnamefont {F.}~\bibnamefont {Hirsch}}, \bibinfo {author} {\bibfnamefont {M.~T.}\ \bibnamefont {Quintino}},\ and\ \bibinfo {author} {\bibfnamefont {N.}~\bibnamefont {Brunner}},\ }\bibfield  {title} {\bibinfo {title} {Quantum measurement incompatibility does not imply bell nonlocality},\ }\href {https://doi.org/10.1103/PhysRevA.97.012129} {\bibfield  {journal} {\bibinfo  {journal} {Phys. Rev. A}\ }\textbf {\bibinfo {volume} {97}},\ \bibinfo {pages} {012129} (\bibinfo {year} {2018})}\BibitemShut {NoStop}%
\bibitem [{\citenamefont {Quintino}\ \emph {et~al.}(2014)\citenamefont {Quintino}, \citenamefont {V\'ertesi},\ and\ \citenamefont {Brunner}}]{quintino16}%
  \BibitemOpen
  \bibfield  {author} {\bibinfo {author} {\bibfnamefont {M.~T.}\ \bibnamefont {Quintino}}, \bibinfo {author} {\bibfnamefont {T.}~\bibnamefont {V\'ertesi}},\ and\ \bibinfo {author} {\bibfnamefont {N.}~\bibnamefont {Brunner}},\ }\bibfield  {title} {\bibinfo {title} {Joint measurability, einstein-podolsky-rosen steering, and bell nonlocality},\ }\href {https://doi.org/10.1103/PhysRevLett.113.160402} {\bibfield  {journal} {\bibinfo  {journal} {Phys. Rev. Lett.}\ }\textbf {\bibinfo {volume} {113}},\ \bibinfo {pages} {160402} (\bibinfo {year} {2014})}\BibitemShut {NoStop}%
\bibitem [{\citenamefont {Uola}\ \emph {et~al.}(2015)\citenamefont {Uola}, \citenamefont {Budroni}, \citenamefont {G\"uhne},\ and\ \citenamefont {Pellonp\"a\"a}}]{uola16}%
  \BibitemOpen
  \bibfield  {author} {\bibinfo {author} {\bibfnamefont {R.}~\bibnamefont {Uola}}, \bibinfo {author} {\bibfnamefont {C.}~\bibnamefont {Budroni}}, \bibinfo {author} {\bibfnamefont {O.}~\bibnamefont {G\"uhne}},\ and\ \bibinfo {author} {\bibfnamefont {J.-P.}\ \bibnamefont {Pellonp\"a\"a}},\ }\bibfield  {title} {\bibinfo {title} {One-to-one mapping between steering and joint measurability problems},\ }\href {https://doi.org/10.1103/PhysRevLett.115.230402} {\bibfield  {journal} {\bibinfo  {journal} {Phys. Rev. Lett.}\ }\textbf {\bibinfo {volume} {115}},\ \bibinfo {pages} {230402} (\bibinfo {year} {2015})}\BibitemShut {NoStop}%
\bibitem [{\citenamefont {Mukherjee}\ and\ \citenamefont {Pan}(2024)}]{mukherjee24}%
  \BibitemOpen
  \bibfield  {author} {\bibinfo {author} {\bibfnamefont {S.}~\bibnamefont {Mukherjee}}\ and\ \bibinfo {author} {\bibfnamefont {A.~K.}\ \bibnamefont {Pan}},\ }\bibfield  {title} {\bibinfo {title} {Constrained measurement incompatibility from generalised contextuality of steered preparation},\ }\href {https://doi.org/10.1088/1367-2630/ad96d8} {\bibfield  {journal} {\bibinfo  {journal} {New Journal of Physics}\ }\textbf {\bibinfo {volume} {26}},\ \bibinfo {pages} {123014} (\bibinfo {year} {2024})}\BibitemShut {NoStop}%
\bibitem [{\citenamefont {Selby}\ \emph {et~al.}(2023)\citenamefont {Selby}, \citenamefont {Schmid}, \citenamefont {Wolfe}, \citenamefont {Sainz}, \citenamefont {Kunjwal},\ and\ \citenamefont {Spekkens}}]{Selby2023}%
  \BibitemOpen
  \bibfield  {author} {\bibinfo {author} {\bibfnamefont {J.~H.}\ \bibnamefont {Selby}}, \bibinfo {author} {\bibfnamefont {D.}~\bibnamefont {Schmid}}, \bibinfo {author} {\bibfnamefont {E.}~\bibnamefont {Wolfe}}, \bibinfo {author} {\bibfnamefont {A.~B.}\ \bibnamefont {Sainz}}, \bibinfo {author} {\bibfnamefont {R.}~\bibnamefont {Kunjwal}},\ and\ \bibinfo {author} {\bibfnamefont {R.~W.}\ \bibnamefont {Spekkens}},\ }\bibfield  {title} {\bibinfo {title} {Contextuality without incompatibility},\ }\href {https://doi.org/10.1103/PhysRevLett.130.230201} {\bibfield  {journal} {\bibinfo  {journal} {Phys. Rev. Lett.}\ }\textbf {\bibinfo {volume} {130}},\ \bibinfo {pages} {230201} (\bibinfo {year} {2023})}\BibitemShut {NoStop}%
\bibitem [{\citenamefont {Singh}\ \emph {et~al.}(2023)\citenamefont {Singh}, \citenamefont {Bhati},\ and\ \citenamefont {Arvind}}]{j.Singh2023}%
  \BibitemOpen
  \bibfield  {author} {\bibinfo {author} {\bibfnamefont {J.}~\bibnamefont {Singh}}, \bibinfo {author} {\bibfnamefont {R.~S.}\ \bibnamefont {Bhati}},\ and\ \bibinfo {author} {\bibnamefont {Arvind}},\ }\bibfield  {title} {\bibinfo {title} {Revealing quantum contextuality using a single measurement device},\ }\href {https://doi.org/10.1103/PhysRevA.107.012201} {\bibfield  {journal} {\bibinfo  {journal} {Phys. Rev. A}\ }\textbf {\bibinfo {volume} {107}},\ \bibinfo {pages} {012201} (\bibinfo {year} {2023})}\BibitemShut {NoStop}%
\bibitem [{\citenamefont {Holevo}(1973)}]{Holevo1973}%
  \BibitemOpen
  \bibfield  {author} {\bibinfo {author} {\bibfnamefont {A.~S.}\ \bibnamefont {Holevo}},\ }\bibfield  {title} {\bibinfo {title} {Bounds for the quantity of information transmitted by a quantum communication channel},\ }\href {http://mathscinet.ams.org/mathscinet-getitem?mr=456936} {\bibfield  {journal} {\bibinfo  {journal} {Problems Inform. Transmission}\ }\textbf {\bibinfo {volume} {9}} (\bibinfo {year} {1973})}\BibitemShut {NoStop}%
\bibitem [{\citenamefont {Frenkel}\ and\ \citenamefont {Weiner}(2015)}]{FW_2015}%
  \BibitemOpen
  \bibfield  {author} {\bibinfo {author} {\bibfnamefont {P.~E.}\ \bibnamefont {Frenkel}}\ and\ \bibinfo {author} {\bibfnamefont {M.}~\bibnamefont {Weiner}},\ }\bibfield  {title} {\bibinfo {title} {Classical information storage in an n-level quantum system},\ }\href {https://doi.org/10.1007/s00220-015-2463-0} {\bibfield  {journal} {\bibinfo  {journal} {Communications in Mathematical Physics}\ }\textbf {\bibinfo {volume} {340}},\ \bibinfo {pages} {563} (\bibinfo {year} {2015})}\BibitemShut {NoStop}%
\bibitem [{\citenamefont {Harrigan}\ and\ \citenamefont {Spekkens}(2010)}]{Harrigan2010}%
  \BibitemOpen
  \bibfield  {author} {\bibinfo {author} {\bibfnamefont {N.}~\bibnamefont {Harrigan}}\ and\ \bibinfo {author} {\bibfnamefont {R.~W.}\ \bibnamefont {Spekkens}},\ }\bibfield  {title} {\bibinfo {title} {Einstein, incompleteness, and the epistemic view of quantum states},\ }\href {https://doi.org/10.1007/s10701-009-9347-0} {\bibfield  {journal} {\bibinfo  {journal} {Foundations of Physics}\ }\textbf {\bibinfo {volume} {40}},\ \bibinfo {pages} {125} (\bibinfo {year} {2010})}\BibitemShut {NoStop}%
\bibitem [{\citenamefont {Leifer}(2014)}]{Leifer_2014}%
  \BibitemOpen
  \bibfield  {author} {\bibinfo {author} {\bibfnamefont {M.}~\bibnamefont {Leifer}},\ }\bibfield  {title} {\bibinfo {title} {Is the quantum state real? an extended review of $\psi$-ontology theorems},\ }\href {https://doi.org/10.12743/quanta.v3i1.22} {\bibfield  {journal} {\bibinfo  {journal} {Quanta}\ }\textbf {\bibinfo {volume} {3}},\ \bibinfo {pages} {67} (\bibinfo {year} {2014})}\BibitemShut {NoStop}%
\bibitem [{\citenamefont {G\"uhne}\ \emph {et~al.}(2023)\citenamefont {G\"uhne}, \citenamefont {Haapasalo}, \citenamefont {Kraft}, \citenamefont {Pellonp\"a\"a},\ and\ \citenamefont {Uola}}]{Guhne2023}%
  \BibitemOpen
  \bibfield  {author} {\bibinfo {author} {\bibfnamefont {O.}~\bibnamefont {G\"uhne}}, \bibinfo {author} {\bibfnamefont {E.}~\bibnamefont {Haapasalo}}, \bibinfo {author} {\bibfnamefont {T.}~\bibnamefont {Kraft}}, \bibinfo {author} {\bibfnamefont {J.-P.}\ \bibnamefont {Pellonp\"a\"a}},\ and\ \bibinfo {author} {\bibfnamefont {R.}~\bibnamefont {Uola}},\ }\bibfield  {title} {\bibinfo {title} {Colloquium: Incompatible measurements in quantum information science},\ }\href {https://doi.org/10.1103/RevModPhys.95.011003} {\bibfield  {journal} {\bibinfo  {journal} {Rev. Mod. Phys.}\ }\textbf {\bibinfo {volume} {95}},\ \bibinfo {pages} {011003} (\bibinfo {year} {2023})}\BibitemShut {NoStop}%
\bibitem [{\citenamefont {Oszmaniec}\ \emph {et~al.}(2017)\citenamefont {Oszmaniec}, \citenamefont {Guerini}, \citenamefont {Wittek},\ and\ \citenamefont {Ac\'{\i}n}}]{ozmaniac17}%
  \BibitemOpen
  \bibfield  {author} {\bibinfo {author} {\bibfnamefont {M.}~\bibnamefont {Oszmaniec}}, \bibinfo {author} {\bibfnamefont {L.}~\bibnamefont {Guerini}}, \bibinfo {author} {\bibfnamefont {P.}~\bibnamefont {Wittek}},\ and\ \bibinfo {author} {\bibfnamefont {A.}~\bibnamefont {Ac\'{\i}n}},\ }\bibfield  {title} {\bibinfo {title} {Simulating positive-operator-valued measures with projective measurements},\ }\href {https://doi.org/10.1103/PhysRevLett.119.190501} {\bibfield  {journal} {\bibinfo  {journal} {Phys. Rev. Lett.}\ }\textbf {\bibinfo {volume} {119}},\ \bibinfo {pages} {190501} (\bibinfo {year} {2017})}\BibitemShut {NoStop}%
\bibitem [{\citenamefont {Patra}\ \emph {et~al.}(2025)\citenamefont {Patra}, \citenamefont {Mukherjee},\ and\ \citenamefont {Pan}}]{akp2025}%
  \BibitemOpen
  \bibfield  {author} {\bibinfo {author} {\bibfnamefont {P.}~\bibnamefont {Patra}}, \bibinfo {author} {\bibfnamefont {S.}~\bibnamefont {Mukherjee}},\ and\ \bibinfo {author} {\bibfnamefont {A.}~\bibnamefont {Pan}, \bibfnamefont {K.}},\ }\bibfield  {title} {\bibinfo {title} {Supplemental material},\ }\href@noop {} {\bibfield  {journal} {\bibinfo  {journal} {Phys. Rev. Lett.}\ } (\bibinfo {year} {2025})}\BibitemShut {NoStop}%
\bibitem [{\citenamefont {D’Ariano}\ \emph {et~al.}(2005)\citenamefont {D’Ariano}, \citenamefont {Presti},\ and\ \citenamefont {Perinotti}}]{D_Ariano_2005}%
  \BibitemOpen
  \bibfield  {author} {\bibinfo {author} {\bibfnamefont {G.~M.}\ \bibnamefont {D’Ariano}}, \bibinfo {author} {\bibfnamefont {P.~L.}\ \bibnamefont {Presti}},\ and\ \bibinfo {author} {\bibfnamefont {P.}~\bibnamefont {Perinotti}},\ }\bibfield  {title} {\bibinfo {title} {Classical randomness in quantum measurements},\ }\href {https://doi.org/10.1088/0305-4470/38/26/010} {\bibfield  {journal} {\bibinfo  {journal} {Journal of Physics A: Mathematical and General}\ }\textbf {\bibinfo {volume} {38}},\ \bibinfo {pages} {5979–5991} (\bibinfo {year} {2005})}\BibitemShut {NoStop}%
\bibitem [{\citenamefont {Perry}\ \emph {et~al.}(2015)\citenamefont {Perry}, \citenamefont {Jain},\ and\ \citenamefont {Oppenheim}}]{Perry2015}%
  \BibitemOpen
  \bibfield  {author} {\bibinfo {author} {\bibfnamefont {C.}~\bibnamefont {Perry}}, \bibinfo {author} {\bibfnamefont {R.}~\bibnamefont {Jain}},\ and\ \bibinfo {author} {\bibfnamefont {J.}~\bibnamefont {Oppenheim}},\ }\bibfield  {title} {\bibinfo {title} {Communication tasks with infinite quantum-classical separation},\ }\href {https://doi.org/10.1103/PhysRevLett.115.030504} {\bibfield  {journal} {\bibinfo  {journal} {Phys. Rev. Lett.}\ }\textbf {\bibinfo {volume} {115}},\ \bibinfo {pages} {030504} (\bibinfo {year} {2015})}\BibitemShut {NoStop}%
\bibitem [{\citenamefont {Harrigan}\ and\ \citenamefont {Rudolph}(2007)}]{Harrigan2007ontologicalmodelsinterpretationcontextuality}%
  \BibitemOpen
  \bibfield  {author} {\bibinfo {author} {\bibfnamefont {N.}~\bibnamefont {Harrigan}}\ and\ \bibinfo {author} {\bibfnamefont {T.}~\bibnamefont {Rudolph}},\ }\href {https://arxiv.org/abs/0709.4266} {\bibinfo {title} {Ontological models and the interpretation of contextuality}} (\bibinfo {year} {2007}),\ \Eprint {https://arxiv.org/abs/0709.4266} {arXiv:0709.4266 [quant-ph]} \BibitemShut {NoStop}%
\bibitem [{\citenamefont {Leifer}\ and\ \citenamefont {Maroney}(2013)}]{Leifer_Marony(2013)}%
  \BibitemOpen
  \bibfield  {author} {\bibinfo {author} {\bibfnamefont {M.~S.}\ \bibnamefont {Leifer}}\ and\ \bibinfo {author} {\bibfnamefont {O.~J.~E.}\ \bibnamefont {Maroney}},\ }\bibfield  {title} {\bibinfo {title} {Maximally epistemic interpretations of the quantum state and contextuality},\ }\href {https://doi.org/10.1103/PhysRevLett.110.120401} {\bibfield  {journal} {\bibinfo  {journal} {Phys. Rev. Lett.}\ }\textbf {\bibinfo {volume} {110}},\ \bibinfo {pages} {120401} (\bibinfo {year} {2013})}\BibitemShut {NoStop}%
\bibitem [{\citenamefont {Pan}(2021)}]{Pan_2021}%
  \BibitemOpen
  \bibfield  {author} {\bibinfo {author} {\bibfnamefont {A.~K.}\ \bibnamefont {Pan}},\ }\bibfield  {title} {\bibinfo {title} {Two definitions of maximally $\psi$-epistemic ontological model and preparation non-contextuality},\ }\bibfield  {journal} {\bibinfo  {journal} {Europhysics Letters}\ }\textbf {\bibinfo {volume} {133}},\ \href {https://doi.org/10.1209/0295-5075/133/50004} {10.1209/0295-5075/133/50004} (\bibinfo {year} {2021})\BibitemShut {NoStop}%
\bibitem [{\citenamefont {Maroney}(2012)}]{Maroney2012}%
  \BibitemOpen
  \bibfield  {author} {\bibinfo {author} {\bibfnamefont {O.~J.~E.}\ \bibnamefont {Maroney}},\ }\href {https://arxiv.org/abs/1207.7192} {\bibinfo {title} {A brief note on epistemic interpretations and the kochen-specker theorem}} (\bibinfo {year} {2012}),\ \Eprint {https://arxiv.org/abs/1207.7192} {arXiv:1207.7192 [quant-ph]} \BibitemShut {NoStop}%
\bibitem [{\citenamefont {Carmeli}\ \emph {et~al.}(2018)\citenamefont {Carmeli}, \citenamefont {Heinosaari},\ and\ \citenamefont {Toigo}}]{Heinosaari2018}%
  \BibitemOpen
  \bibfield  {author} {\bibinfo {author} {\bibfnamefont {C.}~\bibnamefont {Carmeli}}, \bibinfo {author} {\bibfnamefont {T.}~\bibnamefont {Heinosaari}},\ and\ \bibinfo {author} {\bibfnamefont {A.}~\bibnamefont {Toigo}},\ }\bibfield  {title} {\bibinfo {title} {State discrimination with postmeasurement information and incompatibility of quantum measurements},\ }\href {https://doi.org/10.1103/PhysRevA.98.012126} {\bibfield  {journal} {\bibinfo  {journal} {Phys. Rev. A}\ }\textbf {\bibinfo {volume} {98}},\ \bibinfo {pages} {012126} (\bibinfo {year} {2018})}\BibitemShut {NoStop}%
\bibitem [{\citenamefont {Carmeli}\ \emph {et~al.}(2019)\citenamefont {Carmeli}, \citenamefont {Heinosaari},\ and\ \citenamefont {Toigo}}]{CarmeliIW2019}%
  \BibitemOpen
  \bibfield  {author} {\bibinfo {author} {\bibfnamefont {C.}~\bibnamefont {Carmeli}}, \bibinfo {author} {\bibfnamefont {T.}~\bibnamefont {Heinosaari}},\ and\ \bibinfo {author} {\bibfnamefont {A.}~\bibnamefont {Toigo}},\ }\bibfield  {title} {\bibinfo {title} {Quantum incompatibility witnesses},\ }\href {https://doi.org/10.1103/PhysRevLett.122.130402} {\bibfield  {journal} {\bibinfo  {journal} {Phys. Rev. Lett.}\ }\textbf {\bibinfo {volume} {122}},\ \bibinfo {pages} {130402} (\bibinfo {year} {2019})}\BibitemShut {NoStop}%
\bibitem [{\citenamefont {Patra}\ \emph {et~al.}(2024)\citenamefont {Patra}, \citenamefont {Naik}, \citenamefont {Lobo}, \citenamefont {Sen}, \citenamefont {Guha}, \citenamefont {Bhattacharya}, \citenamefont {Alimuddin},\ and\ \citenamefont {Banik}}]{Patra2024}%
  \BibitemOpen
  \bibfield  {author} {\bibinfo {author} {\bibfnamefont {R.~K.}\ \bibnamefont {Patra}}, \bibinfo {author} {\bibfnamefont {S.~G.}\ \bibnamefont {Naik}}, \bibinfo {author} {\bibfnamefont {E.~P.}\ \bibnamefont {Lobo}}, \bibinfo {author} {\bibfnamefont {S.}~\bibnamefont {Sen}}, \bibinfo {author} {\bibfnamefont {T.}~\bibnamefont {Guha}}, \bibinfo {author} {\bibfnamefont {S.~S.}\ \bibnamefont {Bhattacharya}}, \bibinfo {author} {\bibfnamefont {M.}~\bibnamefont {Alimuddin}},\ and\ \bibinfo {author} {\bibfnamefont {M.}~\bibnamefont {Banik}},\ }\bibfield  {title} {\bibinfo {title} {Classical analogue of quantum superdense coding and communication advantage of a single quantum system},\ }\href {https://doi.org/10.22331/q-2024-04-09-1315} {\bibfield  {journal} {\bibinfo  {journal} {{Quantum}}\ }\textbf {\bibinfo {volume} {8}},\ \bibinfo {pages} {1315} (\bibinfo {year} {2024})}\BibitemShut {NoStop}%
\bibitem [{\citenamefont {Ma}\ \emph {et~al.}(2023)\citenamefont {Ma}, \citenamefont {Rambach}, \citenamefont {Goswami}, \citenamefont {Bhattacharya}, \citenamefont {Banik},\ and\ \citenamefont {Romero}}]{Zhonghua2023}%
  \BibitemOpen
  \bibfield  {author} {\bibinfo {author} {\bibfnamefont {Z.}~\bibnamefont {Ma}}, \bibinfo {author} {\bibfnamefont {M.}~\bibnamefont {Rambach}}, \bibinfo {author} {\bibfnamefont {K.}~\bibnamefont {Goswami}}, \bibinfo {author} {\bibfnamefont {S.~S.}\ \bibnamefont {Bhattacharya}}, \bibinfo {author} {\bibfnamefont {M.}~\bibnamefont {Banik}},\ and\ \bibinfo {author} {\bibfnamefont {J.}~\bibnamefont {Romero}},\ }\bibfield  {title} {\bibinfo {title} {Randomness-free test of nonclassicality: A proof of concept},\ }\href {https://doi.org/10.1103/PhysRevLett.131.130201} {\bibfield  {journal} {\bibinfo  {journal} {Phys. Rev. Lett.}\ }\textbf {\bibinfo {volume} {131}},\ \bibinfo {pages} {130201} (\bibinfo {year} {2023})}\BibitemShut {NoStop}%
\bibitem [{\citenamefont {Ding}\ \emph {et~al.}(2024)\citenamefont {Ding}, \citenamefont {Lobo}, \citenamefont {Alimuddin}, \citenamefont {Xu}, \citenamefont {Zhang}, \citenamefont {Banik}, \citenamefont {Bao},\ and\ \citenamefont {Huang}}]{chen_banik2024}%
  \BibitemOpen
  \bibfield  {author} {\bibinfo {author} {\bibfnamefont {C.}~\bibnamefont {Ding}}, \bibinfo {author} {\bibfnamefont {E.~P.}\ \bibnamefont {Lobo}}, \bibinfo {author} {\bibfnamefont {M.}~\bibnamefont {Alimuddin}}, \bibinfo {author} {\bibfnamefont {X.-Y.}\ \bibnamefont {Xu}}, \bibinfo {author} {\bibfnamefont {S.}~\bibnamefont {Zhang}}, \bibinfo {author} {\bibfnamefont {M.}~\bibnamefont {Banik}}, \bibinfo {author} {\bibfnamefont {W.-S.}\ \bibnamefont {Bao}},\ and\ \bibinfo {author} {\bibfnamefont {H.-L.}\ \bibnamefont {Huang}},\ }\bibfield  {title} {\bibinfo {title} {Quantum advantage: A single qubit's experimental edge in classical data storage},\ }\href {https://doi.org/10.1103/PhysRevLett.133.200201} {\bibfield  {journal} {\bibinfo  {journal} {Phys. Rev. Lett.}\ }\textbf {\bibinfo {volume} {133}},\ \bibinfo {pages} {200201} (\bibinfo {year} {2024})}\BibitemShut {NoStop}%
\bibitem [{\citenamefont {Wiesner}(1983)}]{Wiesner1983RAC}%
  \BibitemOpen
  \bibfield  {author} {\bibinfo {author} {\bibfnamefont {S.}~\bibnamefont {Wiesner}},\ }\bibfield  {title} {\bibinfo {title} {Conjugate coding},\ }\href {https://doi.org/10.1145/1008908.1008920} {\bibfield  {journal} {\bibinfo  {journal} {SIGACT News}\ }\textbf {\bibinfo {volume} {15}},\ \bibinfo {pages} {78–88} (\bibinfo {year} {1983})}\BibitemShut {NoStop}%
\bibitem [{\citenamefont {Ambainis}\ \emph {et~al.}(1999)\citenamefont {Ambainis}, \citenamefont {Nayak}, \citenamefont {Ta-Shma},\ and\ \citenamefont {Vazirani}}]{Ambainis1999RAC}%
  \BibitemOpen
  \bibfield  {author} {\bibinfo {author} {\bibfnamefont {A.}~\bibnamefont {Ambainis}}, \bibinfo {author} {\bibfnamefont {A.}~\bibnamefont {Nayak}}, \bibinfo {author} {\bibfnamefont {A.}~\bibnamefont {Ta-Shma}},\ and\ \bibinfo {author} {\bibfnamefont {U.}~\bibnamefont {Vazirani}},\ }\bibfield  {title} {\bibinfo {title} {Dense quantum coding and a lower bound for 1-way quantum automata},\ }in\ \href {https://doi.org/10.1145/301250.301347} {\emph {\bibinfo {booktitle} {Proceedings of the Thirty-First Annual ACM Symposium on Theory of Computing}}},\ \bibinfo {series and number} {STOC '99}\ (\bibinfo  {publisher} {Association for Computing Machinery},\ \bibinfo {address} {New York, NY, USA},\ \bibinfo {year} {1999})\BibitemShut {NoStop}%
\bibitem [{\citenamefont {Ambainis}\ \emph {et~al.}(2002)\citenamefont {Ambainis}, \citenamefont {Nayak}, \citenamefont {Ta-Shma},\ and\ \citenamefont {Vazirani}}]{Ambainis2002RAC}%
  \BibitemOpen
  \bibfield  {author} {\bibinfo {author} {\bibfnamefont {A.}~\bibnamefont {Ambainis}}, \bibinfo {author} {\bibfnamefont {A.}~\bibnamefont {Nayak}}, \bibinfo {author} {\bibfnamefont {A.}~\bibnamefont {Ta-Shma}},\ and\ \bibinfo {author} {\bibfnamefont {U.}~\bibnamefont {Vazirani}},\ }\bibfield  {title} {\bibinfo {title} {Dense quantum coding and quantum finite automata},\ }\href {https://doi.org/10.1145/581771.581773} {\bibfield  {journal} {\bibinfo  {journal} {J. ACM}\ }\textbf {\bibinfo {volume} {49}},\ \bibinfo {pages} {496–511} (\bibinfo {year} {2002})}\BibitemShut {NoStop}%
\bibitem [{\citenamefont {Ghorai}\ and\ \citenamefont {Pan}(2018)}]{ghorai2018}%
  \BibitemOpen
  \bibfield  {author} {\bibinfo {author} {\bibfnamefont {S.}~\bibnamefont {Ghorai}}\ and\ \bibinfo {author} {\bibfnamefont {A.~K.}\ \bibnamefont {Pan}},\ }\bibfield  {title} {\bibinfo {title} {Optimal quantum preparation contextuality in an $n$-bit parity-oblivious multiplexing task},\ }\href {https://doi.org/10.1103/PhysRevA.98.032110} {\bibfield  {journal} {\bibinfo  {journal} {Phys. Rev. A}\ }\textbf {\bibinfo {volume} {98}},\ \bibinfo {pages} {032110} (\bibinfo {year} {2018})}\BibitemShut {NoStop}%
\bibitem [{\citenamefont {Saha}\ \emph {et~al.}(2023)\citenamefont {Saha}, \citenamefont {Das}, \citenamefont {Das}, \citenamefont {Bhattacharya},\ and\ \citenamefont {Majumdar}}]{saha2023pra}%
  \BibitemOpen
  \bibfield  {author} {\bibinfo {author} {\bibfnamefont {D.}~\bibnamefont {Saha}}, \bibinfo {author} {\bibfnamefont {D.}~\bibnamefont {Das}}, \bibinfo {author} {\bibfnamefont {A.~K.}\ \bibnamefont {Das}}, \bibinfo {author} {\bibfnamefont {B.}~\bibnamefont {Bhattacharya}},\ and\ \bibinfo {author} {\bibfnamefont {A.~S.}\ \bibnamefont {Majumdar}},\ }\bibfield  {title} {\bibinfo {title} {Measurement incompatibility and quantum advantage in communication},\ }\href {https://doi.org/10.1103/PhysRevA.107.062210} {\bibfield  {journal} {\bibinfo  {journal} {Phys. Rev. A}\ }\textbf {\bibinfo {volume} {107}},\ \bibinfo {pages} {062210} (\bibinfo {year} {2023})}\BibitemShut {NoStop}%
\bibitem [{\citenamefont {Theurer}\ \emph {et~al.}(2019)\citenamefont {Theurer}, \citenamefont {Egloff}, \citenamefont {Zhang},\ and\ \citenamefont {Plenio}}]{plenioprl19}%
  \BibitemOpen
  \bibfield  {author} {\bibinfo {author} {\bibfnamefont {T.}~\bibnamefont {Theurer}}, \bibinfo {author} {\bibfnamefont {D.}~\bibnamefont {Egloff}}, \bibinfo {author} {\bibfnamefont {L.}~\bibnamefont {Zhang}},\ and\ \bibinfo {author} {\bibfnamefont {M.~B.}\ \bibnamefont {Plenio}},\ }\bibfield  {title} {\bibinfo {title} {Quantifying operations with an application to coherence},\ }\href {https://doi.org/10.1103/PhysRevLett.122.190405} {\bibfield  {journal} {\bibinfo  {journal} {Phys. Rev. Lett.}\ }\textbf {\bibinfo {volume} {122}},\ \bibinfo {pages} {190405} (\bibinfo {year} {2019})}\BibitemShut {NoStop}%
\bibitem [{\citenamefont {Hance}\ \emph {et~al.}(2025)\citenamefont {Hance}, \citenamefont {Ji}, \citenamefont {Matsushita},\ and\ \citenamefont {Hofmann}}]{jonte25}%
  \BibitemOpen
  \bibfield  {author} {\bibinfo {author} {\bibfnamefont {J.~R.}\ \bibnamefont {Hance}}, \bibinfo {author} {\bibfnamefont {M.}~\bibnamefont {Ji}}, \bibinfo {author} {\bibfnamefont {T.}~\bibnamefont {Matsushita}},\ and\ \bibinfo {author} {\bibfnamefont {H.~F.}\ \bibnamefont {Hofmann}},\ }\href {https://arxiv.org/abs/2501.04664} {\bibinfo {title} {External quantum fluctuations select measurement contexts}} (\bibinfo {year} {2025}),\ \Eprint {https://arxiv.org/abs/2501.04664} {arXiv:2501.04664 [quant-ph]} \BibitemShut {NoStop}%
\bibitem [{\citenamefont {Van~Himbeeck}\ \emph {et~al.}(2017)\citenamefont {Van~Himbeeck}, \citenamefont {Woodhead}, \citenamefont {Cerf}, \citenamefont {Garc{\'{i}}a-Patr{\'{o}}n},\ and\ \citenamefont {Pironio}}]{VanHimbeeck2017semidevice}%
  \BibitemOpen
  \bibfield  {author} {\bibinfo {author} {\bibfnamefont {T.}~\bibnamefont {Van~Himbeeck}}, \bibinfo {author} {\bibfnamefont {E.}~\bibnamefont {Woodhead}}, \bibinfo {author} {\bibfnamefont {N.~J.}\ \bibnamefont {Cerf}}, \bibinfo {author} {\bibfnamefont {R.}~\bibnamefont {Garc{\'{i}}a-Patr{\'{o}}n}},\ and\ \bibinfo {author} {\bibfnamefont {S.}~\bibnamefont {Pironio}},\ }\bibfield  {title} {\bibinfo {title} {Semi-device-independent framework based on natural physical assumptions},\ }\href {https://doi.org/10.22331/q-2017-11-18-33} {\bibfield  {journal} {\bibinfo  {journal} {{Quantum}}\ }\textbf {\bibinfo {volume} {1}},\ \bibinfo {pages} {33} (\bibinfo {year} {2017})}\BibitemShut {NoStop}%
\bibitem [{\citenamefont {Tavakoli}\ \emph {et~al.}(2020)\citenamefont {Tavakoli}, \citenamefont {Zambrini~Cruzeiro}, \citenamefont {Bohr~Brask}, \citenamefont {Gisin},\ and\ \citenamefont {Brunner}}]{Tavakoli2020informationally}%
  \BibitemOpen
  \bibfield  {author} {\bibinfo {author} {\bibfnamefont {A.}~\bibnamefont {Tavakoli}}, \bibinfo {author} {\bibfnamefont {E.}~\bibnamefont {Zambrini~Cruzeiro}}, \bibinfo {author} {\bibfnamefont {J.}~\bibnamefont {Bohr~Brask}}, \bibinfo {author} {\bibfnamefont {N.}~\bibnamefont {Gisin}},\ and\ \bibinfo {author} {\bibfnamefont {N.}~\bibnamefont {Brunner}},\ }\bibfield  {title} {\bibinfo {title} {Informationally restricted quantum correlations},\ }\href {https://doi.org/10.22331/q-2020-09-24-332} {\bibfield  {journal} {\bibinfo  {journal} {{Quantum}}\ }\textbf {\bibinfo {volume} {4}},\ \bibinfo {pages} {332} (\bibinfo {year} {2020})}\BibitemShut {NoStop}%
\end{thebibliography}%
\newpage
\begin{widetext}
    \appendix
\section{Detailed proof of Lemma 1: Derivation of maximum quantum success probability  $\mathcal{P}_{Q}^{max}$}\label{qo}
Without loss of generality consider that Alice prepares the following six qubit states 
\ba\label{P1}
\rho_{x0}=\frac{\mathds{I}+ A_x}{2}; \ \ \  \rho_{x1}=\frac{\mathds{I} - A_{x \oplus 2}}{2}, \ \ x \in \{1,2,3\}
\ea
satisfying the restriction of Eq.(4b) of the main text, so that $\half(\rho_{x0}+\rho_{x1})\equiv \rho_1, \ \forall x$.  To satisfy the  restriction of Eq.(4a) of the main text, we need
\ba\label{P2}
\sum_{x=0,1,2}\frac{1}{3}\rho_{x0}=\sum_{x=0,1,2} \frac{1}{3}\rho_{x1}\equiv \rho_{2} 
\ea
which implies that
\ba\label{Cnstrn 2_Observable  form}
 A_0+A_1+A_2=0 
\ea
    
Bob chooses a general three-outcome qubit POVM as given by,
\ba
E_0=\alpha_0\frac{\mathds{I}+ B_0}{2};\ \ E_1= \alpha_1\frac{\mathds{I}+ B_1}{2};\ \ E_2=\alpha_2\frac{\mathds{I}+ B_2}{2} 
\ea
The completeness relation of measurement, $E_0+E_1+E_2=\mathds{I}$ implies
\begin{eqnarray}\label{Condition of Bob's Observables}
        \alpha_0B_0+\alpha_1B_1 +\alpha_2 B_2=0 , \ \ \ \text{and} \ \  \alpha_0+\alpha_1+\alpha_2=2
\end{eqnarray}
Now, using Eqs.\eqref{P1} and \eqref{P2} , the quantum success probability in Eq.(5) of the main text can be explicitly written as,

\ba\label{Avg_Succ_11}
\mathcal{P}_{Q}&&=\frac{1}{6}\Tr[(\rho_{00}+\rho_{21})E_0+ (\rho_{10}+ \rho_{01})E_1+ (\rho_{20}+ \rho_{11})E_2] \nonumber\\
&& = \frac{1}{6}\Tr[\frac{\alpha_0}{4}(2\mathds{I}+ A_0- A_1)( \mathds{I}+B_0)+ \frac{\alpha_1}{4}(2\mathds{I}+ A_1- A_2)(\mathds{I}+B_1) + \frac{\alpha_2}{4} (2\mathds{I}+ A_2- A_0)(\mathds{I}+B_2)].
\ea\label{Avg_Succ_12}
Here we define three normalization constants $\omega_0=||A_0 - A_1||,\omega_1=||A_1 - A_2||,\omega_2=||A_2 - A_0||$, \Big($||\mathcal{N}||$ is Frobenius norm of a matrix $\mathcal{N}$\Big), such that  Eq.\eqref{Avg_Succ_11} can be rewritten as,
\ba
\label{pqq}
\mathcal{P}_{Q}= \frac{1}{3}+ \frac{1}{24}\Tr[\alpha_0\frac{A_0-A_1}{\omega_0}B_0\omega_0 +  \alpha_1 \frac{A_1-A_2}{\omega_1} B_1 \omega_1 +  \alpha_2 \frac{A_2- A_0}{\omega_2}B_2\omega_2].
\ea

It is now evident from Eq.\eqref{pqq} that the maximization of $\mathcal{P}_{Q}$ demands,
\ba
B_0=\frac{A_0 - A_1}{\omega_0}; \  B_1=\frac{A_1 - A_2}{\omega_1}; \ 
 B_2=\frac{A_2-A_0}{\omega_2}. \nonumber
\ea
Since $(B_{y})^{2}=\mathbb{I}$, we get
\ba\label{A8}
\mathcal{P}_{Q} &&\leq \frac{1}{3} + \frac{1}{12}\max[\alpha_0\omega_0 + \alpha_1 \omega_1 + \alpha_2 \omega_2]\\ \nonumber
&&= \frac{1}{3}+ \frac{1}{12}\max\left[\alpha_0\sqrt{\langle 2\mathds{I} -\{A_0,A_1\}\rangle} +\alpha_1\sqrt{\langle 2\mathds{I} - \{A_1,A_2\}\rangle} +\alpha_2\sqrt{\langle 2\mathds{I} -\{A_2,A_0\} \rangle} \right]
\ea

Again form Eq.\eqref{Cnstrn 2_Observable  form} we have,
\begin{eqnarray}\label{cnstrn 1}
    3\mathds{I}+\{A_0 ,A_1\}+\{A_1 ,A_2\}+\{A_2,A_0\}=0
\end{eqnarray} 
Multiply $A_0$ from the right of Eq. \eqref{Cnstrn 2_Observable  form} followed by multiplying same $A_0$ from the left and then summing both equations we obtain, 
\begin{eqnarray} \label{cnstrn 2}
    2\mathds{I} + \{A_0 ,A_1\}+\{A_0 ,A_2\} =0
\end{eqnarray}
Comparing both the Eqs.\eqref{cnstrn 1} -\eqref{cnstrn 2} we get $\langle\{A_1 ,A_2\}\rangle=-1$. Similarly, we can obtain $\langle\{A_0 ,A_1\}\rangle=-1$ and $\langle\{A_0 ,A_2\}\rangle=-1$, following  the constrain Eq.\eqref{Cnstrn 2_Observable  form}. Those relations indicate $\omega_0=\omega_1=\omega_2=1/\sqrt{3}$.  Substituting those in Eq.\eqref{A8} we obtain the bound for success probability as,
\begin{eqnarray}
    \mathcal{P}_{Q}\leq \frac{1}{3}\left(1+ \frac{\sqrt{3}}{2} \right)
\end{eqnarray}

Here, Bob's observables corresponding to the optimal quantum strategy is given by,
\begin{eqnarray}\label{Bob's Observables}
    B_0=\frac{1}{\sqrt{3}}(A_0-A_1);\ \ B_1=\frac{1}{\sqrt{3}}(A_1-A_2);\ \ B_2=\frac{1}{\sqrt{3}}(A_2-A_0)
\end{eqnarray}

Replacing all $B_b$ from Eq.\eqref{Bob's Observables} into Eq.\eqref{Condition of Bob's Observables} and after rearranging we obtain,
\begin{eqnarray}
    (\alpha_0-\alpha_2)A_0+(\alpha_1-\alpha_0)A_1+ (\alpha_2-\alpha_1)A_2=0 \nonumber
\end{eqnarray}
Right away, substituting $A_2=-A_0-A_1$ and multiplying $A_0$ from the right side followed by multiplying the same $A_0$ from the left and then summing both equations we get,
\begin{eqnarray}\label{A12}
    2\mathds{I}(\alpha_0+\alpha_1-2\alpha_2)+\{A_0,A_1\}(2\alpha_1-\alpha_0-\alpha_2)=0 
    \end{eqnarray}
Since $\langle\{A_0 ,A_1\}\rangle=-1$, it follows from Eq.\eqref{A12} that
 $\alpha_0=\alpha_2$. Similarly, substituting $A_0=-A_2-A_1$ and multiplying $A_1$ from the right side followed by multiplying the same $A_1$ from the left and then summing both equations and then using the relation $\langle\{A_2,A_1\}\rangle=-1$ we find $\alpha_0=\alpha_1$. Consolidating these  relations,   we obtain $\alpha_0=\alpha_1=\alpha_2$. Again we have $\alpha_0+\alpha_1+\alpha_2=2$ that straightforwardly indicates $ \alpha_0=\alpha_1=\alpha_2=\frac{2}{3}$. Hence, a suitable choice of states and observables which satisfy these relations can give maximum success probability of the game as $\mathcal{P}^{max}_{Q}=\frac{1}{3}\left(1+ \frac{\sqrt{3}}{2} \right)$.

 One of the examples that satisfy the required properties for maximum success probability is given here: Suppose Alice chooses $A_x= \cos\theta_x \sigma_Z+ \sin\theta_x \sigma_X $ with $\theta_x=\frac{2\pi x}{3}$. 
 From Eq.\eqref{Bob's Observables} we obtain $B_b=\cos\theta_b \sigma_Z- \sin\theta_b \sigma_X $ where $\theta_b=\frac{2\pi}{3}(\frac{1}{4}+b)$. After sating $\alpha_0=\alpha_1=\alpha_2=\frac{2}{3}$ simple calculation of $\mathcal{P}_Q$ gives $\mathcal{P}^{max}_{Q}=\frac{1}{3}\left(1+ \frac{\sqrt{3}}{2} \right)$.

 \section{Proof of Theorem 1: Derivation of maximum success probability in a noncontextual model}\label{AppB}
In Lemma \textbf{2} of the main text, we showed that an ontic state $\lambda$ can strictly be in support of two out of three response functions, $\{\xi(b|E_b,\lambda)\}$ with $b=\{0,1,2\}$. For our purpose, we then divide  the ontic state space ($\Lambda$) into three regions $\Lambda_0,\Lambda_1$ and $\Lambda_2$ depending on the supports of response functions, as depicted in Figure \textbf{1} of the main text. For example, in region $\Lambda_0$ ,  $\xi(0|E_0,\lambda),\xi(2|E_2,\lambda)>0$, but $\xi(1|E_1,\lambda)=0$. Similar arguments hold for the other two regions, as shown in Figure \textbf{1}. Hence, the success probability in Eq.(8) in an ontological model can then be written as,
\ba
    \mathcal{P}_{ont}&&=\frac{1}{6}\sum_{\Lambda_0}\biggr[\mu_{00}\left(\lambda\right)\xi(0|E_0,\lambda)+ \mu_{20}\left(\lambda\right)\xi(2|E_2,\lambda)+ \mu_{11 }\left(\lambda\right)\xi(2|E_2,\lambda)+ \mu_{21}\left(\lambda\right)\xi(0|E_0,\lambda)\biggr] \nonumber\\
    &&+\frac{1}{6}\sum_{\Lambda_1}\biggr[\mu_{00}\left(\lambda\right)\xi(0|E_0,\lambda) +\mu_{10}\left(\lambda\right)\xi(1|E_1,\lambda)\nonumber+\mu_{21}\left(\lambda\right)\xi(0|E_0,\lambda)+\mu_{01}\left(\lambda\right)\xi(1|E_1,\lambda) \biggr] \\
    &&+ \frac{1}{6}\sum_{\Lambda_2}\biggr[\mu_{10}\left(\lambda\right)\xi(1|E_1,\lambda)+\mu_{20}\left(\lambda\right)\xi(2|E_2,\lambda)+\mu_{01}\left(\lambda\right)\xi(1|E_1,\lambda)+\mu_{11}\left(\lambda\right)\xi(2|E_2,\lambda)\biggr] 
\ea
Our task is to maximize $\mathcal{P}_{ont}$ in a noncontextual ontological model, $i.e.$,
\ba
\mathcal{P}_{NC}=\max_{\mu_{xa}(\lambda),\xi(b|E_b,\lambda)} \mathcal{P}_{ont} 
\ea

From Proposition \textbf{2} of the main text we have $0\leq\xi(b|E_b,\lambda)\leq \alpha_b $ with $\sum\limits_{b}\alpha_b =2$. To maximize $P_{ont}$ for the region $\Lambda_0$, we choose the maximum value $\xi(0|E_0,\lambda)=\alpha_0$ and therefore $\xi(2|E_2,\lambda)=1-\alpha_0$. Similarly, for the region $\Lambda_1$ we choose $\xi(1|E_1,\lambda)=\alpha_1$ and then $\xi(0|E_0,\lambda)=1-\alpha_1$, and for the region $\Lambda_2$ we choose $\xi(2|E_2,\lambda)=\alpha_2$ and then $\xi(1|E_1,\lambda)=1-\alpha_2$. One could take any other possible values, but due to the inherent symmetry of the distribution, they will lead equivalent result. We then have,
\ba\label{General NC Success Prob}
    \mathcal{P}_{ont}&&\leq \frac{1}{6}\sum_{\Lambda_0}\alpha_0\left(\mu_{00}\left(\lambda\right)+\mu_{21}\left(\lambda\right) \right) + (1-\alpha_0)\left(\mu_{11}\left(\lambda\right)+\mu_{20}\left(\lambda\right) \right)+ \frac{1}{6}\sum_{\Lambda_1}\alpha_1 \left(\mu_{10}\left(\lambda\right)+\mu_{01}\left(\lambda\right) \right) + (1-\alpha_1)\left(\mu_{21}\left(\lambda\right)+\mu_{00}\left(\lambda\right) \right) \nonumber \\
    &&+ \frac{1}{6}\sum_{\Lambda_2}\alpha_2\left(\mu_{20}\left(\lambda\right)+\mu_{11}\left(\lambda\right) \right) + (1-\alpha_2)\left(\mu_{01}\left(\lambda\right)+\mu_{10}\left(\lambda\right) \right)
\ea 

Note that to satisfy the parity
concealment constraints on Alice's preparation in quantum theory, one needs to satisfy Eq.\eqref{P1} and Eq.\eqref{P2} as a rule of the game. Equivalently, in an noncontextual model corresponding to $\rho_1$ and $\rho_2$ we assign their respective ontic state distributions as $\gamma_1(\lambda)$ and $\gamma_2(\lambda)$, so that

\begin{minipage}{0.45\textwidth}
  \begin{align}
    \gamma_2\left(\lambda\right)&=\frac{1}{3}\left(\mu_{00}\left(\lambda\right) +\mu_{10}\left(\lambda\right)+ \mu_{20}\left(\lambda\right) \right)\label{NC4}\\
    &=\frac{1}{3}\left(\mu_{11}\left(\lambda\right) +\mu_{21}\left(\lambda\right)+ \mu_{01}\left(\lambda\right) \right )\label{NC5}
    \end{align}
\end{minipage}
\begin{minipage}{0.45\textwidth}
    \begin{align}
        \gamma_1\left(\lambda\right)&=\frac{1}{2}\left(\mu_{00}\left(\lambda\right) +\mu_{11}\left(\lambda\right) \right )\label{NC1}\\
        &=\frac{1}{2}\left(\mu_{10}\left(\lambda\right) +\mu_{21}\left(\lambda\right) \right )\label{NC2}\\
        &=\frac{1}{2}\left(\mu_{20}\left(\lambda\right) +\mu_{01}\left(\lambda\right) \right )\label{NC3}
    \end{align}
\end{minipage}

We first derive the maximum noncontextual bound for a special case $\alpha_b=\frac{2}{3} \ \forall b$ followed by the derivation for any arbitrary value of $\alpha_b$.

\subsection{Noncontextual bound for a specific case of $\alpha_b=\frac{2}{3} \ \forall b$}\label{alphab_b=0.5}

As we are interested in obtaining the noncontextual bound corresponding to the optimal quantum strategy, we take $\alpha_b=\frac{2}{3}$ and $\gamma_1(\lambda)=\gamma_2(\lambda)=\mu_{\frac{\mathds{I}}{2}}(\lambda)$ as $\rho_1=\rho_2=\frac{\mathds{I}}{2}$ for $\alpha_b=\frac{2}{3} \ \forall b$.
By replacing $\mu_{01}(\lambda),\mu_{11}(\lambda)$ and $\mu_{21}(\lambda)$ from Eqs.\eqref{NC1},\eqref{NC2} and \eqref{NC3} respectively into Eq.\eqref{General NC Success Prob}, and further simplifying, we obtain,
\ba
    \mathcal{P}_{ont} \leq && \frac{1}{6}\Bigg[\frac{1}{3}\sum_{\Lambda_0}\left(\mu_{00}\left(\lambda\right)+\mu_{20}\left(\lambda\right) \right) +2\sum_{\Lambda_0}\mu_{\frac{\mathds{I}}{2}}\left(\lambda\right)
    + \frac{1}{3}\sum_{\Lambda_1}\left(\mu_{10}\left(\lambda\right)+\mu_{00}\left(\lambda\right) \right) \nonumber\\ &&+2\sum_{\Lambda_1}\mu_{\frac{\mathds{I}}{2}}\left(\lambda\right) +\frac{1}{3}\sum_{\Lambda_2}\left(\mu_{20}\left(\lambda\right)+\mu_{10}\left(\lambda\right) \right) +2\sum_{\Lambda_2}\mu_{\frac{\mathds{I}}{2}}\left(\lambda\right)\Bigg]
\ea
Now, by using Eqs.(\ref{NC4}) and (\ref{NC5}) and by noting the fact that $\sum_{\Lambda}\mu_{\frac{\mathds{I}}{2}}(\lambda)=1$, we get,
\ba
\mathcal{P}_{ont}&& \leq \frac{1}{6}\Bigg[2+\frac{1}{3}\sum_{\Lambda_0}\left(3\mu_{\frac{\mathds{I}}{2}}(\lambda)-\mu_{10}(\lambda)\right)+\frac{1}{3}\sum_{\Lambda_1}\left(3\mu_{\frac{\mathds{I}}{2}}(\lambda)-\mu_{20}(\lambda)\right) +\frac{1}{3}\sum_{\Lambda_2}\left(3\mu_{\frac{\mathds{I}}{2}}(\lambda)-\mu_{00}(\lambda)\right)\Bigg]
\ea
For maximization, we have to take $\sum\limits_{\Lambda_0}\mu_{10}(\lambda)=0$,$\sum\limits_{\Lambda_1}\mu_{20}(\lambda)=0$ and $\sum\limits_{\Lambda_2}\mu_{00}(\lambda)=0$ which gives 

\ba
\mathcal{P}_{NC} \leq  \frac{1}{6}\Bigg[2+\sum_{\Lambda}\mu_{\frac{\mathds{I}}{2}}(\lambda)\Bigg]=\frac{1}{2}
\ea

\subsection{The noncontextual bound for any arbitrary value of $\alpha_b$}\label{arbitary alpha}

To calculate the upper bound to the success probability $\mathcal{P}_{NC}$ of the noncontextual model, we relabel the terms of Eq.\eqref{General NC Success Prob} as,
$\sum \limits_{\Lambda_0}\mu_{xa}\left(\lambda\right)\equiv A_{xa},\sum \limits_{\Lambda_1}\mu_{xa}\left(\lambda\right)\equiv B_{xa},\sum \limits_{\Lambda_2}\mu_{xa}\left(\lambda\right)\equiv C_{xa}, \sum \limits_{\Lambda_i}\gamma_1(\lambda)\equiv p_i,\sum \limits_{\Lambda_i}\gamma_2(\lambda)\equiv q_i$. Thus, the problem reduces to an linear optimization problem of the functional
\begin{equation}
    \mathcal{P}_{ont} \leq \frac{1}{6}\Big[\alpha_0(A_{00}+A_{21} )+(1-\alpha_0)(A_{11}+A_{20})+\alpha_1(B_{10}+B_{01}+(1-\alpha_1)(B_{21}+B_{00})+\alpha_2(C_{20}+C_{11}+(1-\alpha_2)(C_{01}+C_{10}))\Big]
\end{equation}

 with the constraints from Eqs.\eqref{NC1}, \eqref{NC2}, and \eqref{NC3} as, 
 \ba
 &&A_{00}+A_{11}=2p_0, \ \ \ \ 
 A_{21}+B_{21}-C_{10}=1-2p_2, \ \ \ \ 
 A_{20}-B_{01}+C_{20}=1-2p_1\\
&&B_{10}+B_{21}=2P_1, \ \ \ \
 A_{00}+B_{00}-C_{11}=1-2p_2, \ \ \ \
 -A_{20}+B_{01}+C_{01}=1-2p_0\\
 &&C_{20}+C_{01}=2p_2, \ \ \ \
 A_{11}-B_{00}+C_{11}=1-2p_1, \ \ \ \
 -A_{21}+B_{10}+C_{10}=1-2p_0
 \ea
 and the constraints from Eqs.\eqref{NC4} and \eqref{NC5},
 \ba
 &&A_{00}+A_{20}-B_{10}-C_{10}=3q_0-1, \ \  \ A_{11}+A_{21}-B_{01}-C_{01}=3q_0-1,\\
 &&B_{00}+B_{10}-C_{20}-A_{20}=3q_1-1, \ \ \
 B_{21}+B_{01}-A_{11}-C_{11}=3q_1-1,\\
 &&C_{10}+C_{20}-B_{00}-A_{00}=3q_2-1, \ \ \
 C_{01}+C_{21}-B_{11}-A_{11}=3q_2-1
 \ea
 Also,  
\ba
p_0+p_1+p_2=1; \ \ q_0+q_1+q_2=1; \ \ 
\alpha_0+\alpha_1+\alpha_2=2
\ea
We perform a straightforward numerical optimization using \textcolor{blue}{$Mathemetica\ 13.2$} which yields the maximum value  $\mathcal{P}_{NC}=0.58$ for any arbitrary value of $\alpha_b$.
In Figure \textbf{2} of the main text,  we plot maximum quantum and noncontextual  success probabilities against $\alpha_0$ by taking $\alpha_1=\alpha_2$ .

\section{Proof of Lemma \textbf{3}} \label{Appendix Simulation by trine}

    Consider a $n$-faced unbiased coin with outcome $o\in \{0,1,...,n-1\}$, where $n$ is arbitrary, occurring with $\frac{1}{n}$ probability . Each outcome $o$ corresponds to three-outcome qubit POVMs $\mathbb{M}^o \equiv \Big\{E^o_o,E^o_{o\oplus_n\frac{n-1}{2}},E^o_{o\oplus_n\frac{n+1}{2}}\Big\}$, where $E^o_o=\frac{2\cos{\frac{\pi}{n}}}{1+\cos{\frac{\pi}{n}}} \pi_o, E^o_{o\oplus_n\frac{n-1}{2}}=\frac{1}{1+\cos{\frac{\pi}{n}}} \pi_{o\oplus_n\frac{n-1}{2}}, E^o_{o\oplus_n\frac{n+1}{2}}=\frac{1}{1+\cos{\frac{\pi}{n}}}\pi_{o\oplus_n \frac{n+1}{2}}$. It is crucial to note that each $\pi_o$ occurs three times in the measurement contexts $\mathbb{M}^o,\mathbb{M}^{o \oplus_n \frac{n-1}{2}}$ and $\mathbb{M}^{o \oplus_n \frac{n+1}{2}}$ with coefficients $\frac{2\cos{\frac{\pi}{n}}}{1+\cos{\frac{\pi}{n}}}$, $\frac{1}{1+\cos{\frac{\pi}{n}}}$ and $\frac{1}{1+\cos{\frac{\pi}{n}}}$, respectively. The coefficients then sum up to $2$, and since the probability of performing $\mathbb{M}^o$ is $\frac{1}{n}$, the coefficient $\frac{2}{n}$ of each $\pi_k$ of POVM $E_k$ is obtained. Therefore, instead of measuring $\mathbb{M}_n$ if one measures $\mathcal{S}_n\equiv\{M^o\}$ where $o\in \{0,1,...,n-1\}$ occurs with $\frac{1}{n}$, one gets the same statistics as $\mathbb{M}_n$. 

\subsection{Simulation of five-outcome POVM of  Selby \emph{et al.,}\cite{Selby2023} by five three-outcome extremal qubit POVMs}\label{Appendix Simulation of Selby's Measurement}
To show contextuality without incompatibility \cite{Selby2023}, authors considered a five-outcome measurement which is a special case of our Lemma \textbf{3} of the main text where measurement is given by, $\mathbb{M}_{5}\equiv\{E_k=\frac{2}{5}\pi_k\}$, here $\pi_k= \frac{1}{2}\left( \mathds{I}+ \cos{(\frac{2\pi}{5}k)}\sigma_z + \sin{(\frac{2\pi}{5} k)}\sigma_x\right):k=0,1,2,3,4$.
So, using our Lemma \textbf{3}, a particular set of the trine POVMs $\mathcal{S}_M\equiv\{M^0,M^1,M^2,M^3,M^4\}$, can simulate measurement $\mathbb{M}_{5}$, where,
\ba 
M^0\equiv \{   E_0^0=h_0\pi_0, 
      E_2^0=h_1\pi_2,
      E_3^0=h_1\pi_3\}\\
    M^1 \equiv \{E_1^1=h_0\pi_1,
    E_3^1=h_1\pi_3,
    E_4^1=h_1\pi_4\}\\
        M^2 \equiv \{E_2^2=h_0\pi_2, 
    E_4^2=h_1\pi_4,   E_0^2=h_1\pi_0\}\\
    M^3\equiv \{ E_3^3=h_0\pi_3,
    E_0^3=h_1\pi_0,
    E_1^3=h_1\pi_1\}  \\
    M^4\equiv \{    E_4^4=h_0\pi_4,
    E_1^4=h_1\pi_1,
    E_2^4=h_1\pi_2\}
    \ea 
With, $h_0=\frac{2\cos{\frac{\pi}{5}}}{1+\cos{\frac{\pi}{5}}}$ and $h_1=\frac{1}{1+\cos{\frac{\pi}{5}}}$.

Instead of performing $\mathbb{M}_{5}$, if one takes a five-face unbiased coin and tosses it and if he gets $o$ he measures $M^o \ \forall o \in \{0,1,2,3,4\}$. After his measurements $\{M^o\}$ if he gets $E_k^o$, he reports final output as $E_k$ for the measurement $\mathbb{M}_{5}$. So, probability of getting each $o$ is $\frac{1}{5}$ and probability of getting $E_k$ is,
\ba
p(k|\rho,\mathbb{M})&&=\sum_o p(o) \Tr[E_k^o \rho]=\frac{1}{5}\Tr[\sum_o E_k^o \rho]=\Tr[E_k \rho]
\ea
So, without measuring $\mathbb{M}_{5}$ one can get the same result from post-processing of a set of measurements $\mathcal{S}_M$. We show in Appendix \ref{Incompatibility of S_M} that the set $\mathcal{S}_M$ is incompatible.

\subsection{Proof of measurement incompatibility of any two measurements in the set \texorpdfstring{$\mathcal{S}_M$}{S	extunderscore M}}\label{Incompatibility of S_M}
Compatibility of the set $\mathcal{S}_M$ implies all measurements in this set can be jointly measured. If we can show that any pair from this set is not jointly measurable, it is sufficient to claim that $\mathcal{S}_M$ is \emph{incompatible}. We use the result of Carmeli $\it{et\  al.}$ \cite{CarmeliIW2019,Heinosaari2018} to examine the incompatibility of two measurements $M^0$ and $M^1$ . Their result when mapped to the scenario of our work can be described as follows. As argued in \cite{CarmeliIW2019,Heinosaari2018}, to test the incompatibility of two measurements, it is sufficient to study the state discrimination problem with pre- and post-measurement information. 

Suppose, Alice prepares a state ensemble $\mathcal{E}=\{e_0,e_1,e_2,e'_0,e'_1,e'_2\}$ from a uniformly random distribution and sends to Bob and his task is to discriminate between the states. Also, consider that there are two partitions, $\mathcal{E}_0=\{e_0,e_1,e_2\}$ and $\mathcal{E}_1=\{e'_0,e'_1,e'_2\}$. If before Bob's measurement Alice informs him about the partition from which the state is chosen, then the average success probability for this state discrimination task with prior information is defined as $P^{prior}_{guess}(\mathcal{E};M^0,M^1)$, where $M^0$ and $ M^1$ are the choice of Bob's POVM for the partitions  $\mathcal{E}_0$ and $\mathcal{E}_1$ respectively. 
Now, consider another scenario where Alice sends the state to Bob, and after his measurement is performed, she informs him the information about the partitions. In such a case protocol, the maximum average success probability of the state discrimination with post-information of the partition is defined as $P^{post}_{guess}(\mathcal{E})$.

According to ${Theorem \ 2}$ of \cite{CarmeliIW2019}, two measurements $M^0$ and $M^1$ are incompatible \emph{if and only if} there exists a partitioned ensemble $\mathcal{E}$ such that $P^{prior}_{guess}(\mathcal{E};M^0,M^1)>P^{post}_{guess}(\mathcal{E})$. For our purpose, let us consider the state ensemble $\mathcal{E}$ as follows.

\ba 
\mathcal{E}_0: &&\ \ e_0= \frac{\mathds{I}+\sigma_z}{2}; 
    e_1= \frac{\mathds{I}-\cos{\frac{\pi}{5}}\sigma_z+\sin{\frac{\pi}{5}}\sigma_x}{2} ;
    e_2=\frac{\mathds{I}-\cos{\frac{\pi}{5}}\sigma_z-\sin{\frac{\pi}{5}}\sigma_x}{2}\\
\mathcal{E}_1:&& \ \ e'_0= \frac{\mathds{I}+\cos{\frac{2\pi}{5}}\sigma_z+\sin{\frac{2\pi}{5}}\sigma_x}{2};
    e'_1= \frac{\mathds{I}-\cos{\frac{\pi}{5}}\sigma_z-\sin{\frac{\pi}{5}}\sigma_x}{2};
    e'_2= \frac{\mathds{I}+\cos{\frac{2\pi}{5}}\sigma_z-\sin{\frac{2\pi}{5}}\sigma_x}{2}
    \ea

It is straightforward to understand that with the prior information of $\mathcal{E}_0$ and $\mathcal{E}_1$ the optimal strategy will be to use the measurement $M^0$ for $\mathcal{E}_0$ and $M^1$ for $\mathcal{E}_1$. Therefore, the average success probability of guessing the states is, 
\begin{eqnarray}
    P^{prior}_{guess}=\frac{1}{6}\left( \Tr[e_0 E_0+ e_1 E_2+ e_2 E_3] + \Tr[e'_0 E_1+e'_1 E_3 +e'_2 E_4] \right)=\frac{2}{3}
\end{eqnarray}

On the other hand using post-measurement information \cite{Heinosaari2018}, to calculate the maximum success probability of state discrimination, one of the approaches is to make a new state ensemble $\mathcal{F}$ where one considers all possible combinations of states between the two partitions $\mathcal{E}_0$ and $\mathcal{E}_1$ followed by discriminating those states of the new ensemble. In this scenario the state ensemble $\mathcal{F}$ is given by,
\begin{eqnarray}\label{New-state-ensemble}
    \mathcal{F}=\{\mathcal{F}_{i,j}\}:\mathcal{F}_{ij}=\frac{e_i+e'_j}{2}\ \ \forall e_i\in \mathcal{E}_1,e_j \in \mathcal{E}_2
\end{eqnarray}

We calculate the sum of all the maximum eigenvalues of $\mathcal{F}_{i,j}$, which turns out to be $\iota^{max}=0.943$. Therefore, maximum guessing probability of state discrimination of the ensemble $\mathcal{F}$ is,
\begin{eqnarray}
    P_{guess}(\mathcal{F})\leq\frac{2}{9} \iota^{max}
\end{eqnarray}
It is crucial to note that  each state ($e_i$ or $e'_j$) of a partition ($\mathcal{E}_0\  or\ \mathcal{E}_1$) contributes three times in the state ensemble $\mathcal{F}$. Now, with post-measurement information of partition,  optimal guessing probability of state discrimination of the ensemble $\mathcal{E}$ would be $P^{post}_{guess}(\mathcal{E})=3 P_{guess}(F)$. Finally, we have $P^{post}_{guess} \leq  3 \times \frac{2}{9} \times 0.943=0.629$.

It is then evident that for the measurements $M^0$ and $ M^1$,  there exists a state ensemble $\mathcal{E}$ for which $P^{post}_{Guess}(\mathcal{E})<P^{prior}_{guess}(\mathcal{E};M^0,M^1)$. Therefore, directly applying Theorem $\textbf{2}$ of \cite{CarmeliIW2019}, we establish that the measurements $M^0$ and $ M^1$ are incompatible. It is straightforward to show that any two measurements $M^o$ and $M^{o'\neq o}$ in the set $\mathcal{S}_M$ are incompatible.

\section{PROOF OF THEOREM 2: Maximum success probability using cbit communication aided  by unbounded shared randomness}\label{optimization of bit comm}
Let us first analyze a particular strategy when Alice sends one cbit to Bob without using shared randomness. Assume that before starting the game, Alice and Bob agree upon an encoding ($\mathds{E}_1$) and decoding ($\mathds{D}_1$) strategy as given below.\\

\textbf{Encoding strategy ($\mathds{E}_1$)}: For the input $xa$, Alice prepares '$0$' with probability $p_{xa}$ and '$1$' with probability $1-p_{xa}$.\\

 \textbf{Decoding strategy ($\mathds{D}_1$)}: If Bob receives '$0$' he outputs $0$ with probability $s_0$, $1$ with probability $s_1$ and $2$ with probability $s_2$. Similarly, if he receives $"1"$, he outputs $b$ with probability $r_b, \ \forall b\in \{0,1,2\}$ with $\sum_b s_b=\sum_b r_b=1$. An explicit description of the strategy is given in the following table.\\
 
\begin{center}
\begin{tabular}{||c| c c c c c c| c ||} 
 \hline
 Message & 00 & 01 & 10 & 11 & 20 & 21 & Decoding \\ [0.5ex] 
 \hline\hline
 0 & $p_{00}$ & $p_{01}$ & $p_{10}$ &$p_{11}$ & $p_{20}$ & $p_{21}$ & $s_0\rightarrow0$;\ $s_1 \rightarrow 1$;\ $s_2\rightarrow 2$ \\ 
 \hline
 1 & 1-$p_{00}$ & 1-$p_{01}$ & 1-$p_{10}$ & 1-$p_{11}$ & 1-$p_{20}$ & 1-$p_{21}$ & $r_0\rightarrow0$;\ $r_1 \rightarrow 1$;\ $r_2\rightarrow 2$ \\[1ex] 
 \hline 
\end{tabular}
\end{center}
In order to satisfy the restriction Eq.(4a) and Eq.(4b) of the game, Alice's preparation adhere,
 \begin{eqnarray}\label{c1}
&&\frac{1}{3}\left(p_{00}+p_{10}+p_{20}\right)= \frac{1}{3}\left(p_{01}+p_{11}+p_{21}\right)=\kappa_1\\
\label{c2}
&&\frac{1}{2}\left(p_{00}+p_{11}\right)=\frac{1}{2}\left(p_{01}+p_{20}\right)=\frac{1}{2}\left(p_{10}+p_{21}\right)=\kappa_2
\end{eqnarray}
where $\kappa_1, \kappa_2\in [0,1]$. It is simple to check from Eqs. (\ref{c1}) and (\ref{c2}) that $\kappa_1=\kappa_2$.

Now, consider that Alice and Bob share an unbounded shared randomness denoted by the variable $\tau$ with probability $p(\tau)$ and use this correlation to win the game. A general approach could be to fix an encoding strategy ($\mathds{E}_{\tau}$) and decoding strategy ($\mathds{D}_{\tau}$)  corresponding to each $\tau$. Hence, the average success probability ($\mathcal{P}_c$) of the game due to Alice's one $c$-bit communication aided with shared randomness  becomes,
\begin{eqnarray}\label{c3}
\mathcal{P}_c&&=\sum_{\tau}p(\tau)\mathcal{P}(\mathds{E}_{\tau},\mathds{D}_{\tau})\nonumber \\
    &&=\frac{1}{6}\sum_{\tau}p(\tau)([p_{00}(\tau)+p_{21}(\tau)]s_0(\tau) +[p_{11}(\tau)+p_{20}(\tau)]s_2(\tau)+ [p_{01}(\tau)+p_{10}(\tau)]s_1(\tau) \nonumber \\
    &&+  [2-p_{00}(\tau)-p_{21}(\tau)]r_0(\tau) +[2-p_{11}(\tau)-p_{20}(\tau)]r_2(\tau)+ [2-p_{01}(\tau)-p_{10}(\tau)]r_1(\tau))
\end{eqnarray}
where $p_{xa}(\tau)$ is the short-hand notation for the probability $p(0|xa,\tau)$ corresponding to Alice's message $'0'$ for the input $(x,a)$ and the shared variable $\tau$ and similarly $s_{b}(\tau)$ and $r_{b}(\tau)$ are the probabilities of Bob's output $'0'$ and $'1'$ respectively given $\tau$. Using Eq.\eqref{c2} we eliminate $p_{01}(\tau),p_{11}(\tau),p_{21}(\tau)$ and substitute $s_2(\tau)=1-s_1(\tau)-s_0(\tau)$, $r_2(\tau)=1-r_1(\tau)-r_0(\tau)$ in Eq.\eqref{c3}, we finally obtain,
 \begin{equation}
     \mathcal{P}_c=\frac{1}{3} + \frac{1}{6} \sum_{\tau}p(\tau)\Big[\big(2p_{00}(\tau)-p_{10}(\tau)-p_{20}(\tau)\big)\big(s_0(\tau)-r_0(\tau)\big) +  
     \big(p_{10}(\tau)+p_{00}(\tau)-2p_{20}(\tau)\big)\big(s_1(\tau)-r_1(\tau)\big)\Big]
 \end{equation}
Subsequently, using Eq.\eqref{c1} we get,
 \begin{equation}
 \label{cc}
     \mathcal{P}_c=\frac{1}{3}+ \frac{1}{2}\sum_{\tau} p(\tau)\Big[\big(p_{00}(\tau)-\kappa_2(\tau) \big)\big(s_0(\tau)-r_0(\tau)\big) +  \big(\kappa_2(\tau)-p_{20}(\tau)\big)\big(s_1(\tau)-r_1(\tau)\big)\Big]
 \end{equation}
 From Eq. (\ref{cc}), it can be seen that there are four possibilities of values, (I)  $s_0(\tau)\geq r_0(\tau),   s_1(\tau) \geq r_1(\tau) $, (II) $s_0(\tau) \geq r_0(\tau),   s_1(\tau)\leq r_1(\tau) $, (III) $s_0(\tau)\leq r_0(\tau),   s_1(\tau) \geq r_1(\tau) $, and (IV) $s_0(\tau)\leq r_0(\tau),   s_1(\tau)\leq r_1(\tau) $. Let us now analyze the maximum values for all four cases.

For case (I), maximization requires $r_0(\tau)=r_1(\tau)= 0$ and $s_0(\tau)= 1$, leading to $s_1(\tau)=0$. Eq. (\ref{cc}) then reduces to 
\begin{equation}
 \label{cc1}
     \mathcal{P}^{(I)}_c=\sum_{\tau\in (I)} p(\tau)\Big[p_{00}(\tau)-\kappa_2(\tau) \Big]
 \end{equation}
For maximization, we further take $p_{00}(\tau)=1$ and following Eq.\eqref{c2} the minimum value of $\kappa_2(\tau)$ can be $\half$. Putting these together, we get 
\begin{eqnarray}
    \mathcal{P}_c^{(I)}=\sum_{\tau \in (I)}\half p(\tau) 
\end{eqnarray}
Similarly, for other three cases it is straightforward to check that the result remains the same as above $i.e.$, $\mathcal{P}_c^{(II)}=\sum\limits_{\tau \in (II)}\half p(\tau)$, $\mathcal{P}_c^{(III)}=\sum\limits_{\tau \in (III)}\half p(\tau)$ and $\mathcal{P}_c^{(IV)}=\sum\limits_{\tau \in (IV)}\half p(\tau)$. Considering all those relations, we find the maximum value, $(\mathcal{P}_c)^{max}$ of $\mathcal{P}_c$ to be 
\begin{eqnarray}
\left(\mathcal{P}_c\right)^{max}=\frac{1}{3} + \half \sum_{\tau}\half p(\tau) =\frac{7}{12}\approx 0.58
\end{eqnarray} 
Therefore, the maximum success probability achieved for the communication game for Alice's one $c$-bit communication assisted with unbounded shared randomness is $0.58$. Then the maximum success probability (as in Lemma. \textbf{1} of the main text) of the communication game using qubit communication  surpasses the one achieved with $c$-bit, establishing the quantum supremacy in HFW scenario.

\section{Coherence as the non-classical feature of measurements viable for contextuality without incompatibility}\label{Appendix Coherence}

The resource theory of operations based on their ability to detect coherence of quantum states was proposed in \cite{plenioprl19}. Since coherence is a basis-dependent quantity, before defining the properties of POVMs that are capable of detecting the coherence of states, let us fix the basis states as $\{\ket{\zeta_{i}}\}$. Any state that is a convex mixture of those basis states is called incoherent states and represented as,

\begin{equation}
    \rho_{inc}=\sum_{i}p_{i} \ketbra{\zeta_{i}}{\zeta_{i}}, \ \ where, \sum_{i}p_{i}=1.
\end{equation}
The incoherent states are produced by completely dephasing operation as $\rho_{inc}= \Lambda (\rho)$, where $\Lambda$ is a coherence destroying map. Now, a POVM $\{E_{n}^{free}\}$ is said to be free if and only if it cannot detect the coherence of any state such that $tr[\Lambda(\rho)E_{n}^{free}]=tr[\rho E_{n}^{free}] \ \forall \rho, n $. It is argued in \cite{plenioprl19} that the form that is necessary and sufficient for a POVM to be free is

\begin{equation}
\label{eq:incoherent}
  E_{n}^{free}= \sum_{i} \alpha_{n,i} \ketbra{\zeta_{i}}{\zeta_{i}} \ \ \forall n
\end{equation}

Keeping these definitions in mind, let us now consider the POVM $(\mathbb{M})$ used in this paper, whose elements are of the form $E_{b}=\alpha_{b}\ketbra{\psi_b}{\psi_b}$, with $b\in\{0,1,2\}$ and $\sum_{b}\alpha_{b}=2$. In order to examine whether $E_{b}$ can be written in the form of Eq. \eqref{eq:incoherent} we fix a particular basis as $\{\ket{\psi_{0}},\ket{\psi_{0}^{\perp}}\}$. It is straightforward to see that only at the extreme cases of $\alpha_{0}=0$  and $\alpha_{0}=1$, the POVM $E_{b}$ can be written in the form of Eq.\eqref{eq:incoherent}. However, these two extreme cases are trivial as they represent the two-outcome sharp qubit measurements. For any other values of $\alpha_b$, the measurement remains a three-outcome qubit measurement and is not a free POVM. Therefore, we argue that in the absence of measurement incompatibility, the ability of our three-outcome qubit POVM to detect the state coherence, is the non-classicality  that provides the quantum advantage. 

\end{widetext}

\end{document}